\documentclass[aps,12pt,superscriptaddress,tightenlines,floatfix]{revtex4-1}
\usepackage{amsmath,amssymb,amsthm}
\usepackage[T1]{fontenc}
\usepackage[tt=false]{libertine}
\usepackage[scaled=0.83]{beramono}
\usepackage[libertine]{newtxmath}
\usepackage{graphicx}
\usepackage{bm}
\usepackage{physics}
\usepackage{booktabs}
\usepackage{tabularx}
\usepackage{mathtools}
\mathtoolsset{showonlyrefs,showmanualtags}
\usepackage[caption=false]{subfig}
\usepackage{calrsfs}

\newcommand{\nn}{\nonumber}
\renewcommand{\phi}{\varphi}
\newcommand{\C}{\mathbb{C}}
\newcommand{\CP}{\mathbb{CP}}
\newcommand{\N}{\mathbb{N}}

\newcommand{\Z}{\mathbb{Z}}
\DeclareMathOperator{\GL}{GL}
\providecommand{\smallabs}[1]{\lvert#1\rvert}
\DeclarePairedDelimiter\ceil{\lceil}{\rceil}

\theoremstyle{plain}
\newtheorem{theorem}{Theorem}[section]
\newtheorem{definition}[theorem]{Definition}
\newtheorem{conjecture}[theorem]{Conjecture}

\allowdisplaybreaks

\begin{document}
\title{Decoherence and the Classes of Maximally Entangled States}
\author{Roman V. Buniy}
\email{roman.buniy@gmail.com}
\affiliation{Schmid College of Science, Chapman University, Orange, CA 92866, USA}
\author{Robert P. Feger}
\email{robert.feger@gmail.com}
\affiliation{Deutscher Wetterdienst, Frankfurter Str. 135, 63067 Offenbach am Main, Germany}
\author{Thomas W. Kephart} 
\email{tom.kephart@gmail.com}
\affiliation{Department of Physics and Astronomy, Vanderbilt University, Nashville, TN 37235, USA}
\date{\today}

\begin{abstract}
Self-interactions and interaction with the environment tend to push quantum systems toward states of maximal entanglement.
This is a definition of decoherence.
We argue that these maximally entangled states fall into the well-defined classes that can be uniquely described by the values of certain entanglement invariants.
After discussing these ideas we present examples of maximally entangled states for a number of generic systems, construct compact states in the most entangled classes for tripartite systems, and suggest how they may be constructed for other $n$-partite systems.
We study random walks through the space of entanglement classes to see how decoherence might work in practice.
\end{abstract}

\maketitle

\section{Introduction}
\label{section_introduction}

The physics of entanglement and decoherence \cite{EPR,Zeh,Zurek:1981xq,Zurek:1982ii,Zurek:2003zz,Schlosshauer:2003zy,Joos,Schlosshauer,Bengtsson} is taking on an ever more important role in contemporary physics. 
The success of quantum devices, and in particular of quantum computers \cite{Bruzewicz,Arute,Zhong,Wu,Ebadi}, depends on understanding and control of quantum aspects of these systems.
Assuming we can arrange quantum states of a system with sufficient precision, they are subject to disruptions by all sorts of phenomena---from interactions with their local environment to interactions with cosmic rays and cosmic background radiation.
Such states are typically fragile and difficult to protect from unwanted agents of change.
One can shield quantum states from some, but not all, of these interactions, and they tend to move states toward decoherence, by which we mean that an initial carefully prepared isolated quantum state of a system becomes ever more entangled with itself and its environment.
Without rigorous attention to shielding, decoherence proceeds very quickly to a state of maximal entanglement of the system and the environment. 

One profitable avenue of exploration of the world of entangled states is their taxonomy which proceeds by organizing them in such sets that members of the same set have certain important common features that unite them, while members of distinct sets differ from each other in ways that highlight their physical distinctiveness.
This approach is very similar to the original rank-based scientific classification that was introduced long ago by Carl Linnaeus \cite{Linnaeus} and applied to a biological classification of organisms into kingdoms, classes, orders, genera and species.
Other branches of science have also benefited from taxonomy. 

As a result of numerical experiments with entangled states in a wide variety of systems, we have observed their many properties and collected extensive entanglement data, which leads us to several conjectures, some of which we have proved and some others that perhaps can be proved with more effort.
We suggest that understanding of quantum entanglement can benefit from a detailed classification, i.e., by defining the entanglement taxonomy to organize the overwhelming amount of information that has been and will be discovered about entangled states.
This perspective can also be found in \cite{Bengtsson}.

Classification of entangled states has been studied for over two decades \cite{Dur,Verstraete,Klyachko,Miyake1,Luque1,Miyake2,Briand,Luque2,Toumazet,Wallach,Lamata,Buniy:2010yh,Buniy:2010zp} and has its mathematical foundations in linear algebra \cite{Gelfand} and classical invariant theory \cite{Olver}.
In view of the large variety of entangled states (even the full description of all types of which is presently beyond our reach), we will focus here on the study of the maximally entangled states and attempt to find universal properties that describe them.

Since quantum states transform into each other by coordinate transformations, to describe physically distinct states, we need to group states together into sets in such a way that states from the same set can be transformed into each other, but there are no transformations that relate states from different sets.
This leads to all quantum states falling into sets of equivalence classes, where the equivalence relation is induced by transformation of bases of vector spaces.
It turns out that this results in the infinite number of equivalence classes.
However, if one further relates these classes by certain criteria and partition them into sets of equivalence classes (or, equivalently, forms equivalence classes of equivalence classes), then this leads to a finite set of equivalence classes. 

Here is a slightly more detailed description of the formation of the above equivalence classes and the resulting entanglement classification which is based on algebraic invariants that originate in linear algebra \cite{Buniy:2010yh,Buniy:2010zp}.
We consider an $n$-partite system $S$ consisting of $n$ subsystems $\{S_i\}_{i=1}^n$.
(We often use the shorter term $n$-partite even for the cases where the longer terms are available: bipartite ($n=2$), tripartite ($n=3$), quadripartite ($n=4$), quintipartite ($n=5$), hexapartite ($n=6$), septempartite ($n=7$), octapartite ($n=8$), etc.)
Each subsystem $S_i$ is described by a vector space $V_i$ and the corresponding space for the full system $S$ is the tensor product $V=\otimes_{i=1}^n V_i$.
Each $V_i$ is invariant under the general linear group $\GL(V_i)$ and consequently $V$ is invariant under the direct product group $G=\times_{i=1}^n \GL(V_i)$.
By considering all vectors into which a given vector $v\in V$ is transformed by $G$, we arrive at the equivalent class $C(v)$ of $v$.
Proceeding with all $v\in V$, we partition $V$ into the set of equivalence classes $\mathcal{C}$.
To define invariants, we partition the system $S$ into two subsystems $T$ and $T'$, which leads to the corresponding product $V=W\otimes W'$, and consider linear maps between $W$ and $W'$.
Each of these maps are fully characterized by the dimensions of their kernels and images, and since the maps depend on a given $v\in V$, these dimensions become algebraic invariants associated with each state $v$. 
Finally, to get the complete information about any given state $v\in V$, we consider all possible intersections between subspaces of $V$ and compute the invariants of the associated linear maps.
This leads to the complete set of algebraic invariants for every $v\in V$.
By this construction, the values of any algebraic invariant is the same for every vector in the equivalence class $C(v)$ of a vector $v$.
For a full discussion see \cite{Buniy:2010yh,Buniy:2010zp}.

The plan of the present paper is as follows.
In Section \ref{section_fragility_of_quantum_states_and_decoherence} we give a definition of maximally entangled states and elaborate on its meaning, in particular because other definitions for similar concepts exist in the literature.
This naturally leads to the definition of maximally entangled classes of states. 
Among various entanglement invariants, we define the principal entanglement invariant which encodes the amount of entanglement of the highest degree available for a given system and plays the crucial role in the distinction between maximally and non-maximally entangled states and classes.
Section \ref{section_examples_of_maximally_entangled_states} contains numerous examples of maximally entangled states and classes.
Here we consider hypercubic and hypercuboid systems separately; these are systems of qudits that have, respectively, the same and arbitrary dimensions.
This separation allows us to look at maximally entangled states of hypercubic systems with more analytical details.
Additionally, numerical data on entanglement for the simplest hypercubic systems of $n$ qubits is easier to study than any other $n$-partite systems, and the difference between these two types of systems increases exponentially with $n$.
We use the numerical data from Sections \ref{section_hypercubic_systems} and \ref{section_hypercuboid_systems} to formulate several conjectures in Section \ref{section_general_systems} about the principal entanglement invariant and the shortest lengths of maximally entangled states.
Then, to get an idea of the rate of decoherence, we study random walks through the space of entanglement classes in Section \ref{section_a_random_walk_through_the_space_of_entanglement_classes}.
The basic idea here is to allow evolution of states in the form of random changes of their coordinates (which might be caused by generic unspecified interactions within the system or between the system and the environment) and observe the resulting transition of states between different entanglement classes of the system.
This provides a view of the ``entanglement landscape.''
For example, we learn the measures of connectedness of different classes, i.e., how easy it is for a state to move from one class to another in the sense of the expected value for the number of steps in the above random walks that it takes to do so.
We end with a discussion and conclusions in Section \ref{section_discussion_and_conclusions}, where we mention some important and clear examples that can be used to study entanglement and subsequent decoherence.

\section{Fragility of Quantum States and Decoherence}
\label{section_fragility_of_quantum_states_and_decoherence}

In general, the details of what constitutes maximally entangled states (MESs) are not known because they are quite complex, but in small systems (or small systems plus environments) there are many cases where they can be precisely described, and that is what we plan to discuss in this paper.
We argue that all the MESs of any given system form the class of states which is uniquely defined by values of their algebraic entanglement invariants.
Numerous examples of such states for various multipartite systems of various dimensions will be provided.
We will work with generic systems and will not need nor discuss details of interactions, time scales, partial traces, reduced density matrices and other concepts that describe the evolution toward decoherence since all these topics can be found in the literature. 
(For a review see \cite{Schlosshauer} and references therein, where they are adequately discussed and the original references are given.)

\subsection{Invariants, Maximally Entangled States and Classes}
\label{section_invariants_maximally_entangled_states_and_classes}

For a system of $n$ qudits with Hilbert spaces $\{V_i\}_{i=1}^n$, the Hilbert space of the combined system is the tensor product space $V=\otimes_{i=1}^n V_i$.
Let $d_i=\dim{V_i}$ for $1\le i\le n$, so that $\dim{V}=\prod_{i=1}^n d_i$.
Since $\dim{V}\ge d_{\textrm{min}}^n$, where $d_{\textrm{min}}=\min{(d_1,\dotsc,d_n)}$, it follows that $\dim{V}\ge 2^n$ for $d_{\textrm{min}}\ge 2$, and this exponential dependence of $\dim{V}$ on $n$ is responsible for difficulties in studying $n$-partite systems with large $n$.

We choose a basis $E_i=\{e_{i,j}\}_{j=1}^{d_i}$ of $V_i$ for each $i$ and use the resulting tensor product basis $E=\otimes_{i=1}^n E_i$ of $V$ to write an arbitrary element $v\in V$ in terms of its coordinates $\{v_{j_1,\dotsc,j_n}\}$ according to
\begin{align}
  v &=\sum_{j_1=1}^{d_1} \dotsb \sum_{j_n=1}^{d_n} v_{j_1,\dotsc,j_n} e_{1,j_1}\otimes \dotsb \otimes e_{n,j_n}.
  \label{v}
\end{align}
For a general state $v\in V$, the number of terms in \eqref{v} equals $\dim{V}$.
However, states look different in different bases, which implies that, in particular, the entanglement of a state cannot be deduced from only the number of terms in its expansion \eqref{v}.
For example, it may happen that there exists a choice of the basis $E'=\otimes_{i=1}^n E'_i$ with $E'_i=\{e'_{i,j_i}\}_{j_i=1}^{d_i}$, where $e'_{i,1}=\sum_{j_i=1}^{d_i} c_{i,j_i}e_{i,j_i}$ and $\{c_{i,j_i}\}$ are  complex numbers satisfying $\sum_{j_i=1}^{d_i}\smallabs{c_{i,j_i}}^2 =1$ for $1\le i\le n$, such that $v$ in \eqref{v} can be written in the factorized form with only one term, for example, $v=e'_{1,1}\otimes \dotsb \otimes e'_{n,1}$.
The existence of such a factorized form is the necessary and sufficient condition for $v$ to be an unentangled state.
All other states (those that cannot be brought to a factorized form by any choice of the basis) are entangled.

Although the set of all unentangled states and the set of all entangled sets are both infinite, it is intuitively clear that, in any non-trivial system, the set of entangled states is significantly larger then the set of unentangled states.
It is then useful to partition the set of all entangled states into subsets based on some characteristics of degree of entanglement of states.
The simplest of such characteristics is the partiteness, which approximately corresponds to the number of subsystems that are inextricably joined, and such states are $k$-partite entangled, where $2\le k\le n$.
This characterization can be made precise by using the algebraic entanglement invariants, which are determined by different types of partitioning of a system into subsystems.
Representing an $n$-partite system by the set $\{1,\dotsc,n\}$ and its subsystems by the corresponding subsets of this set, we need to consider all partitions of $\{1,\dotsc,n\}$ into $m$ subsets, where $2\le m\le n$. 
This leads to a finite set of invariants $\mathcal{N}_D$ for any $n$-partite system with the dimensions $D=(d_1,\dots,d_n)$.

An important feature of $\mathcal{N}_D$ is that the number of invariants $\smallabs{\mathcal{N}_D}=\nu_{\smallabs{D}}$ depends only on $\smallabs{D}=n$.
With this property, computations of entanglement invariants is simplified in the sense that one needs to compute analytic forms of the invariants only once for a set of any $n$-partite systems.
After these forms are established, the difference between any two $n$-partite systems with different dimensions is due to the values that the invariants take.

The number $\nu_n$ rapidly increases with $n$; see Table \ref{table_numbers_of_invariants} for $2\le n\le 7$.
\begin{table}[htpb]
  \centering
  \addtolength{\tabcolsep}{2pt}    
  \begin{tabular}{ccccccc}
    \toprule
    $n$ & $2$ & $3$ & $4$ & $5$ & $6$ & $7$ \\
    \midrule
    $\nu_n$ & $1$ & $4$ & $19$ & $167$ & $11,747$ & $12,160,646$ \\
    \bottomrule
  \end{tabular}
  \addtolength{\tabcolsep}{-2pt}    
  \caption{The numbers of invariants $\smallabs{\mathcal{N}_D}=\nu_{\smallabs{D}}$, where $\smallabs{D}=n$, for any $n$-partite system for $2\le n\le 7$.}
  \label{table_numbers_of_invariants}
\end{table}
(Note that $\nu_n$ is somewhat unexpectedly related to a sequence in \cite{OEIS}, which is the number of maximal intersecting anti-chains of subsets of $\{1,\dotsc,n\}$.)
The numbers $\nu_n$ for $n\ge 8$ are currently unknown.
The full sets of $3$- and $4$-partite mappings can be found in \cite{BKx}, and the full sets of $5$- and $6$-partite mappings have yet to be presented.

\begin{table}[htpb]
  \centering
  \addtolength{\tabcolsep}{2pt}    
  \begin{tabular}{cccccc}
    \toprule
    $n$ & $2$ & $3$ & $4$ & & $5$ \\
    \midrule
    $\smallabs{\mathcal{C}_{(2,\dotsc,2)}}$ & $3$ & $7$ & $83$ & & $\gtrsim 2\times 10^5$ \\
    \bottomrule
  \end{tabular}
  \addtolength{\tabcolsep}{-2pt}    
  \caption{The numbers of entanglement classes $\smallabs{\mathcal{C}_{(2,\dotsc,2)}}$ for $n$ qubits, where $2\le n\le 5$.}
  \label{table_numbers_of_entanglement_classes}
\end{table}
 
The number of all classes $\smallabs{\mathcal{C}_D}$ of an $n$-partite system with the dimensions $D=(d_1,\dots,d_n)$, which is determined by the allowed values of the full set of invariants $\mathcal{N}_D$, grows rapidly with the number of subsystems $n$ and with the dimensions $D$.
We list several exact values and an approximate lower bound of $\smallabs{\mathcal{C}_D}$ for $n$ qubits in Table \ref{table_numbers_of_entanglement_classes}, and note that no useful results for $n\ge 6$ qubits are known.

Among all the $\smallabs{\mathcal{C}_D}$ classes of a given system with the dimensions $D$, of particular interest is the maximally entangled class.
This class consists of maximally entangled states and every such state has definite mathematical properties distinguishing it from any state outside of the maximally entangled class.
Physical properties of such states mirror their mathematical properties.
We define maximally entangled states for arbitrary dimensions $D$ in such a way that for systems with small partiteness and dimensionality (for example, for $2$ and $3$ qubits) our definition gives the well-known maximally entangled states, which are extensively studied in the literature.

\begin{definition}
A maximally entangled state (MES) of an $n$-partite system with the dimensions $D=(d_1,\dotsc,d_n)$ is any state $v$ for which every invariant in the set $\mathcal{N}_D(v)$ takes its least possible value.
The set of all maximally entangled states is the maximally entangled class (MEC) $C_D$.
\label{definition_maximally_entangled_states_and_classes}
\end{definition}

We make the following comment about the existence of MESs and MECs.
Since every invariant in $\mathcal{N}_D$ is a non-negative integer, there always exist states for which any given invariant in the set $\mathcal{N}_D$ is minimized.
The minimization used in Definition \ref{definition_maximally_entangled_states_and_classes}, however, requires that every invariant in $\mathcal{N}_D$ takes its minimum value.
It might happen that for some $D$ such a global minimization does not result in any state of the system.
We have found MESs and MECs in all numerous examples that we studied, but of course, such numerical studies are not substitutions for the proof of existence of MESs and MECs for arbitrary $D$. 
As a result, we are content with the situation with regards to Definition \ref{definition_maximally_entangled_states_and_classes} that it is possible that for some $D$ maximally entangled states do not exists and the resulting maximally entangled class is empty, but we consider this possibility highly unlikely.

As an example, there are $7$ entanglement classes for $3$ qubits and as we will see, there exists a class (called $C_6$ below) that we can identify with the MEC, for which all the invariants take their least (zero) values.

\subsection{Typical States}
\label{section_typical_states}

We will now define typical states of any system in such a way that they behave in the following manner.
If the system is left to interact with itself and with the environment, then, after a sufficiently long time, the state of the system inevitably and arbitrarily closely approaches one such typical state.
To keep the definition as general as possible, we do not specify the exact relation between the time and proximity of this approach because these parameters depend on the specific details of the system (its composition, dimensions, forms and strengths of interactions within the system and between the system and the environment, etc.).
We replace precise details of the system and environment with the assumption that, for any sufficiently large system and environment with non-trivial interactions, the complexity of the time evolution of a state of the system is such that we might as well assume that it undergoes the process of randomization that equally affects all parts of the state.
This leads us to the following definition.

\begin{definition}
A typical state $v$ of an $n$-partite system with the dimensions $D=(d_1,\dotsc,d_n)$ is a state in which the coordinates $\{v_{j_1,\dotsc,j_n}\}$ in \eqref{v} are taken as independent random variables uniformly distributed over any finite region of $\C^{\dim{V}}$ of dimension $\dim{V}$.
\label{definition_typical_states}
\end{definition}

It is now widely agreed that decoherence is due to uncontrolled  entanglement of a system with its environment to the point where only ``classical'' information can be extracted from the system.
Any original entanglement of the system becomes obscured by the overwhelming amount of entanglement with its environment.
The system evolves quantum-mechanically with no need to impose the collapse of its wave function.
The end point of this evolution is one of infinitely many completely randomized states, i.e., a state where $\{v_{j_1,\dotsc,j_n}\}$ is a set of independent random variables in $\C^{\dim{V}}$.
This is precisely a typical state; see Definition \ref{definition_typical_states}.
For our purposes it is sufficient to assume that details of the initial state of the system do not matter and that the interaction with the environment forces the system to eventually visit every (no matter how small) region of the combined Hilbert space with nonzero probability, which is a type of quantum ergodicity.
We can go further and require that the corresponding probability distribution function is uniform, but for this we need a finite region in $\C^{\dim{V}}$ as its support.
The nature of the algebraic entanglement invariants is such that the choice of this finite region is arbitrary as long as its dimension equals $\dim{V}$.
In view of the equivalence of states $v$ and $cv$ for any $v\in V$ and any nonzero $c\in\C$, it is natural to choose the projective space $\CP^{\dim{V}}$ for this region.
In practice, however, this is not necessary and any other finite region will work for the purpose of computing entanglement invariants for completely randomized states.

The connection between typical states and maximally entangled states is given by the following conjecture.

\begin{conjecture}
Any typical state is almost surely an MES.
Equivalently, for any typical state $v$ and any MES $\tilde{v}$ of the same system, the equality $\mathcal{N}_D(v)=\mathcal{N}_D(\tilde{v})$ holds with probability $1$. 
\label{conjecture_typical_states_are_MESs}
\end{conjecture}

The classes of states of a system are each associated with  specific sets of mappings of one set of subsystems into another, and the mappings are connected with the degree of entanglement.
Roughly speaking, the simplest maps correspond to various types of bipartite entanglement, more complicated maps to tripartite entanglement, and so forth, all within the same system.
The most complicated maps of an $n$-component systems correspond to $n$-partite entanglement, and hence we expect an MES to lie in one of these classes.

We have found that the randomized states of a system of $n$-components tend to fall into a particular $n$-partite class.
This class has the highest level of $n$-partite mappings.
Relative to this class, all other classes form a set of measure zero. Furthermore, as we know from investigating many examples, this MEC has the minimum allowed set of values for its invariants.
There is one invariant, which we call the principal entanglement invariant, associated with the highest level of $n$-partite mappings which always takes on its minimum value for each MES.
In addition, all other subsidiary invariant values also take on their minimum values for each MES.
Hence in this sense, each MES is the ground state of the $n$-partite system.
(We note that when the subsidiary invariants are at their minimum values, the principal invariant can still be above its minimum value.
This occurs for highly entangled states that are still not a part of the MEC.)

Consider a system under which we have complete control to place it in any state we wish, meaning that we can choose all coefficients in the state $v$ in \eqref{v}.
We can then find the class to which this state belongs by calculating all its invariants.
Conversely, if we are able to measure all coefficients of an arbitrary state $v$, then we can use that information to calculate its invariants and hence determine the class to which it belongs.
In practice, neither the construction of the state nor the measurement of the coefficients are accomplished easily, but given the coefficients we can always find the class. 

Let us next assume that our system is in some state for which we have obtained its coefficients in some way.
If it is a random state, then the coefficients $\{v_{j_1,\dotsc,j_n}\}$ are a set of random complex numbers.
Now taking an ensemble of such random states, we are interested in finding the probability that the system is in a given entanglement class.
I.e., we can ask what is the most likely class in which the system finds itself.
This invariably turns out to be the MEC as defined above.
According to the hierarchy of states from most to least entangled, states of lesser and lesser entanglement are more and more unlikely to be found in an ensemble of random states.

As we go from the MEC to the least entangled class, i.e., the class to which the unentangled factorizable states belong, the invariant values increase, although in a complicated way.
However, in the MEC they always have their minimum value, while in the completely unentangled class they take on their maximum values.
(We are not considering the vacuum case here.)

Starting with a system in any of its least entangled states as an initial condition,  random interactions between parts of the system and between the system and the environment will quickly cause the system to decohere into an MES.
As mentioned above, the rate at which this happens has already been thoroughly discussed and estimated elsewhere (for a review and references see \cite{Schlosshauer}), but the MESs themselves have not been explored as thoroughly.
That is what we undertake here by way of formal analytic results where possible, but supported and supplemented by numerous examples.

To begin our study of decoherence towards maximally entangled states, we first list the most important mappings (for details see \cite{Buniy:2010zp}) that we will find to be associated with the most entangled classes for several types of systems.
These are related to the principal entanglement invariants defined below.

\subsection{Principal Entanglement Invariants}
\label{section_principal_entanglement_invariants}

We now return to the discussion in Section \ref{section_invariants_maximally_entangled_states_and_classes}, where we represented subsystems of an $n$-partite system by the set $\{1,\dotsc,n\}$ and considered all partitions of this set into $m$ subsystems, where $2\le m\le n$, which eventually led to a finite set of invariants $\mathcal{N}_n$ for any $n$-partite system.
Of particular interest to us here are the types of partitions that lead to the $n$-partite entanglement, which is the highest order of entanglement available to any $n$-partite system.
Among these, the subsets that generate the principal entanglement invariant are given by 
\begin{align}
    X_n &=\{J_1,\dotsc,J_n\}, 
\end{align}    
where
\begin{align}
    J_i &=\{1,\dotsc,i-1,i+1,\dotsc,n\}.
\end{align}
We can represent the subset $J_i$ as the $(n-2)$-dimensional simplex that is dual to the vertex $\{i\}$ in the $(n-1)$-dimensional simplex representing the set $\{1,\dotsc,n\}$; see Figure \ref{figure_simplices_4} for the case $n=4$.
Although this representation is very simple, it is nevertheless useful because there are similar representations of the  entanglement invariants of lower degrees by simplices of lower dimensions.

\begin{figure}[htpb]
\includegraphics[width=0.75\textwidth]{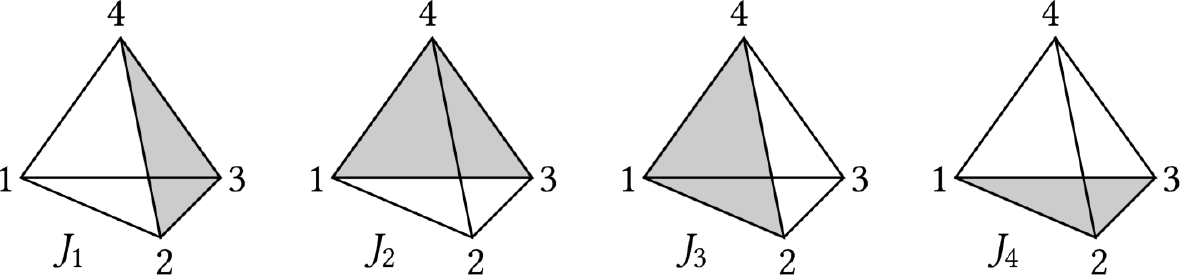}
\caption{Simplex representations of the $4$ subsystems $J_1=\{2,3,4\}$, $J_2=\{1,3,4\}$, $J_3=\{1,2,4\}$, $J_4=\{1,2,3\}$ of the $4$-partite system $\{1,2,3,4\}$ that generate the principal entanglement invariant.}
\label{figure_simplices_4}
\end{figure}

For an $n$-partite system with the dimensions $D=(d_1,\dotsc,d_n)$, we denote the principal entanglement invariant by $N_D$ and note that it is one of the invariants in the set $\mathcal{N}_D$.
This invariant can take values in the range $0\le N_D\le \dim{V}$, but we note that not all of these values can be realized for a given $D$.
When the degree of the $n$-partite entanglement increases, the corresponding value of $N_D$ decreases.
As a result, a maximally $n$-partite entangled state (and, more generally, a maximally entangled state) will have the least possible value of $N_D$.
This leads us to Definition \ref{definition_maximally_entangled_states_and_classes}.
See \cite{Buniy:2010yh,Buniy:2010zp} for a comprehensive discussion.

For a general $n$-partite state, however, the condition that all $n$ subsystems are entangled is only a necessary, but not a sufficient condition for the state to be maximally entangled.
To verify that a given state $v\in V$ is an MES, one needs to compute the set of all entanglement invariants $\mathcal{N}_D(v)$ for this state as well as the set of all entanglement invariants $\mathcal{N}_D(\tilde{v})$ for a state $\tilde{v}\in V$ for which all coordinates $\{\tilde{v}_{j_1,\dotsc,j_n}\}$ are independent random variables in $\C^{\dim{V}}$.  
Since any such $\tilde{v}$ is (extremely likely to be) an MES, we conclude that $v$ is also an MES if and only if $\mathcal{N}_D(v)=\mathcal{N}_D(\tilde{v})$.

Our numerous examples in the following sections will show that the principal invariant of states in each MEC takes its minimum values.
Monte Carlo simulations in which states evolve in such a way that the principal invariant for them reaches the smallest allowed value for each set of dimensions $D$, will show us one aspect of how decoherence works in practice.

\section{Examples of Maximally Entangled States}
\label{section_examples_of_maximally_entangled_states}

For computational reasons, we consider in this section various $n$-partite systems only with $3\le n\le 6$ and various dimensions $D\,{=}\,(d_1,\dotsc,d_n)$.
In view of the large number of invariant features of these systems that characterize their entanglement properties, we restrict ourselves to only  a few invariant characteristics associated with every such system.

The first such characteristic for a given $n$-partite system with the dimensions $D$ is the smallest value of the principal entanglement invariant $N_D$.
This invariant characterizes the purely $n$-partite entanglement properties of the system in the sense that the $k$-partite entanglement properties for all $2\le k\le n-1$ are already described by other algebraic entanglement invariants from the set $\mathcal{N}_D$.
As a result, the states with the smallest values of $N_D$ are the $n$-partite entangled to the fullest possible extent.
These are MESs and they belong to an MEC $C_D$, as described above.

With each state $v$ we associate the number $L(v)$ (which we call the length of $v$) that equals the number of nonzero coordinates $\{v_{j_1,j_2,\dotsc,j_n}\}$ in the expansion \eqref{v} of $v$.
The quantity $L(v)$ changes under a general change of the basis of $v$, and therefore $L(v)$ cannot be an invariant property of $v$.
However, for a given class $C$, we can define the numbers
\begin{align}
&L_{\textrm{min}}(C) =\min{\{L(v)\colon v\in C\}}, \label{l_min} \\
&L_{\textrm{max}}(C) =\max{\{L(v)\colon v\in C\}}, \label{l_max}
\end{align}
which do not depend on a choice of the basis for $v$.
We note that $L_{\textrm{max}}(C)=d_1\dotsb d_n$ for any $C$ different from $C=\{0\}$ since we can always choose basis transformations which increase the number of nonzero coordinates, unless the number has already reached its largest possible value, which is $\dim{V}=d_1 \dotsb d_n$.
The situation with the number $L_{\textrm{min}}$ is very different since the smaller the number $L(v)$, the harder it is to find basis transformations which will decrease the number of nonzero coefficients in the expansion of $v$ while keeping $v$ in the same equivalence class $C$.
This naturally leads to a nontrivial value of $L_{\textrm{min}}(C)$.
For a given set of dimensions $D$, we have a unique MEC $C_D$ and with it the smallest number of nonzero coefficients, which we denote simply $L_D=L_{\textrm{min}}(C_D)$. 
It is clear that any $v$ satisfying $L(v)<L_D$ cannot belong to $C_D$ for the given $D$.
This number $L_D$ is the second (after $N_D$) characteristic of every MEC of interest to us in this study.

The numbers $L_{\textrm{min}}(C)$, where $C$ is not an MEC, will make an indirect appearance in our studies of random walks through the space of entanglement classes in Section \ref{section_a_random_walk_through_the_space_of_entanglement_classes}.
In such random walks, states of a system undergo transformations which might change the number of nonzero coefficients in the expansion \eqref{v}, which might change the entanglement class to which an evolving state belongs.
The relevant generalization of $L_D$ is then $L_{\textrm{min}}(C)$ in \eqref{l_min}, where $C$ is not restricted to be an MEC (as it is for $L_D$). 

\subsection{Hypercubic Systems}
\label{section_hypercubic_systems}

The problem of entanglement classification of a given system involves finding all its entanglement classes.
It can be solved by computing algebraic invariants of each state of the system and assigning a class to which a state belongs based on the values of its invariants. 
An $n$-partite system with the dimensions $D=(d_1,\dotsc,d_n)$, where $n$ is small and the dimensions $d_1,\dotsc,d_n$ of all subsystems are small, has relatively few entanglement classes.
The problem of entanglement classification of such systems is relatively easy; see the results in \cite{Buniy:2010zp,BKx}.

We are interested here in finding short symmetric states belonging to given MECs for general sets of dimensions $D$.
By symmetric states we mean states that have certain symmetries of their coordinates in a given basis, and we prefer states with simple symmetries.
By short states we mean states whose lengths are  close to the smallest length $L_D$ for a given set of dimensions $D$, and therefore less than $dim(V)$.
The goal here is not to find states with the smallest length $L_D$ because such states do not always have easily identifiable simple symmetries.
Such short MESs are a useful tool for comparing states of different degrees of entanglement, for example, states obtained during random walks through the space of entanglement classes as described below in Section \ref{section_a_random_walk_through_the_space_of_entanglement_classes}.

\begin{table}[htpb]
  \centering
  \addtolength{\tabcolsep}{2pt}    
  \begin{tabular}{ccc}
    \toprule
    $D$ & $\mathcal{N}_D$ & $N_D$ \\\midrule
    $(d,d,d)$ & $(0^3,N_D)$ & $\max{(0,d^3-3d^2+2)}$ \\
    $(d,d,d,d)$ & $(0^{18},N_D)$ & $d^4-4d^2+3$ \\
    $(2,2,2,2,2)$ & $(0^{110},7^{10},0^{30},N_D,0^{16})$ & $2^5-5\times 2^2+4=16$ \\
    $(d,d,d,d,d)$ & $(\dotsc,N_D,\dotsc)$ & $d^5-5d^2+4$ \\
    $(2,2,2,2,2,2)$ & $(0^{1792},36^{15},0^{1740},18^{20},0^{960},N_D,0^{7219})$ & $2^6-6\times 2^2+5=45$ \\
    \bottomrule
  \end{tabular}
  \addtolength{\tabcolsep}{-2pt}    
  \caption{The invariants $\mathcal{N}_D$ and the principal invariants $N_D$ for the MECs of hypercubic systems with the dimensions $D$.
  The notation $a^m$ in $\mathcal{N}_D$ means that the invariant value $a$ is repeated $m$ times.
  The numbers of invariants $\smallabs{\mathcal{N}_D}$ agree with Table \ref{table_numbers_of_invariants}.
  For example, there are $11,747$ invariants for six qubits.}
  \label{table_invariants_for_hypercubic_systems}
\end{table}

To simplify the analysis, we start with the dependence on only one dimension among $n$ dimensions in $D$ and temporarily restrict our interest to a subsets of such systems, the hypercubic systems, by which we mean $n$-partite systems of $n$ $d$-qudits, so that $D=(d,\dotsc,d)$.
Table \ref{table_invariants_for_hypercubic_systems} lists our results for the MECs of hypercubic systems.
We obtained these results for various ranges of dimensions $2\le d\le d_\text{max}$ for different $n$, where (for computational reasons) the upper bound of the dimensions $d_\text{max}$ decreases when $n$ increases.
We see that all invariants in $\mathcal{N}_D$ except for the principal invariant $N_D$ are zero only for $2\le n\le 4$, but the principal invariant still follows the pattern of the previous cases even for $n\ge 5$ where there are other nonzero invariants besides $N_D$.
Note that all the values for $N_D$ for varying $d$ in Table \ref{table_invariants_for_hypercubic_systems} vanish for $d=1$, which is consistent with the corresponding system $D=(1,\dotsc,1)$ being a singleton for each $n$.
These results lead us to the following conjecture.

\begin{conjecture}
  \label{conjecture_hypercubic_systems}
  Every state in the MEC for an $n$-partite system with the dimensions $D=(d,\dotsc,d)$, where $d\ge 2$, has the following values of the $2$-partite entanglement invariants $n_{D,i}$, $1\le i\le n$ and the principal entanglement invariant $N_D$:
  \begin{align}
    n_{D,i}=0, \ 1\le i\le n, \ N_D= \max{(0,d^n-nd^2+n-1)}, \ d\ge 2.
    \label{}
  \end{align}
\end{conjecture}

Note that Conjecture \ref{conjecture_hypercubic_systems} specifies all the algebraic invariants in  $\mathcal{N}_D$ only for $n=3$, as there are unspecified invariants for $n\ge 4$.
Compare  Tables \ref{table_numbers_of_invariants} and \ref{table_invariants_for_hypercubic_systems}.

It is convenient to represent any state $v\in \otimes_{i=1}^n V_i$ by an $n$-dimensional diagram that includes the point $(j_1,\dotsc,j_n)$ of the $d_1\times \dotsb \times d_n$ lattice if and only if $v_{j_1,\dotsc,j_n}\not =0$, no matter what this nonzero values is.
The state with all nonzero coefficients in the expansion \eqref{v} is represented by a diagram with every point of the $d_1\times \dotsb \times d_n$ lattice included.
Note that if we truncate this lattice to the size $d'_1\times \dotsb \times d'_n$, where $d'_i\le d_i$, $1\le i\le n$, then the state \eqref{v} becomes
\begin{align}
  v' &=\sum_{j_1=1}^{d'_1} \dotsb \sum_{j_n=1}^{d'_n} v_{j_1,\dotsc,j_n} e_{1,j_1} \otimes \dotsb \otimes e_{n,j_n}.
  \label{v_prime}
\end{align}
If \eqref{v} is an MES for the system $D=(d_1,\dotsc,d_n)$, then \eqref{v_prime} is an MES for the system $D=(d'_1,\dotsc,d'_n)$.
 
The length of a state $v$ is the number $L(v)$ of nonzero coefficients in a given expansion \eqref{v}, and such a state is represented by the above diagram with the number of lattice points equal to $L(v)$.
Although it is interesting to find MESs with the smallest number of nonzero coefficients $L_D$, this is a particularly challenging problem for arbitrary dimensions $D$.
Instead, a solution of a simpler problem is still useful: find MESs with the numbers of nonzero coefficients that are close to being the smallest but with such coefficients arranged in a symmetric manner. 
As a specific example of a set of such states, consider the following state of the $3$-partite cubic system:
\begin{align}
  v =\sum_{l=1}^d e_{1,l}\otimes e_{2,l}\otimes e_{3,l} +\sum_{l=2}^d (e_{1,1}\otimes e_{2,l}\otimes e_{3,l} +e_{1,l}\otimes e_{2,1}\otimes e_{3,l} +e_{1,l}\otimes e_{2,l}\otimes e_{3,1}).
  \label{v_s_3}
\end{align}
We call this state $S_3$-symmetric because it is invariant under all elements of the symmetric group $S_3$, which in this case are permutations of the set of the vector space labels $S=\{1,2,3\}$ (the set of first indices in $\{e_{i,j_i}\}$). 
Figure \ref{444MES} is a diagram for the state \eqref{v_s_3} with $d=4$, where we additionally connected the points on the 2D face diagonals and the 3D body diagonal for visualization.
If we truncate the $d\times d\times d$ diagram to a smaller $d'\times d'\times d'$ diagram, where $d'<d$, the corresponding state is an MES for the truncated system $D=(d',d',d')$; this is similar to the truncation of the full diagram that corresponds to a state with all nonzero coefficients.
On the other hand, if we truncate a symmetric hypercubic MES to a hypercuboid state, it does not necessarily reduce to an MES.
For example, an MES \eqref{v_s_3} with $d=5$ for $D=(5,5,5)$ reduces to an MES for $D=(4,5,5)$, but further reduction to $D=(3,5,5)$ does not yield an MES.
This should be contrasted with truncation of the fully saturated diagram, which always transforms an MES into an MES.
 
\begin{figure}[htpb]
\includegraphics[width=0.55\textwidth]{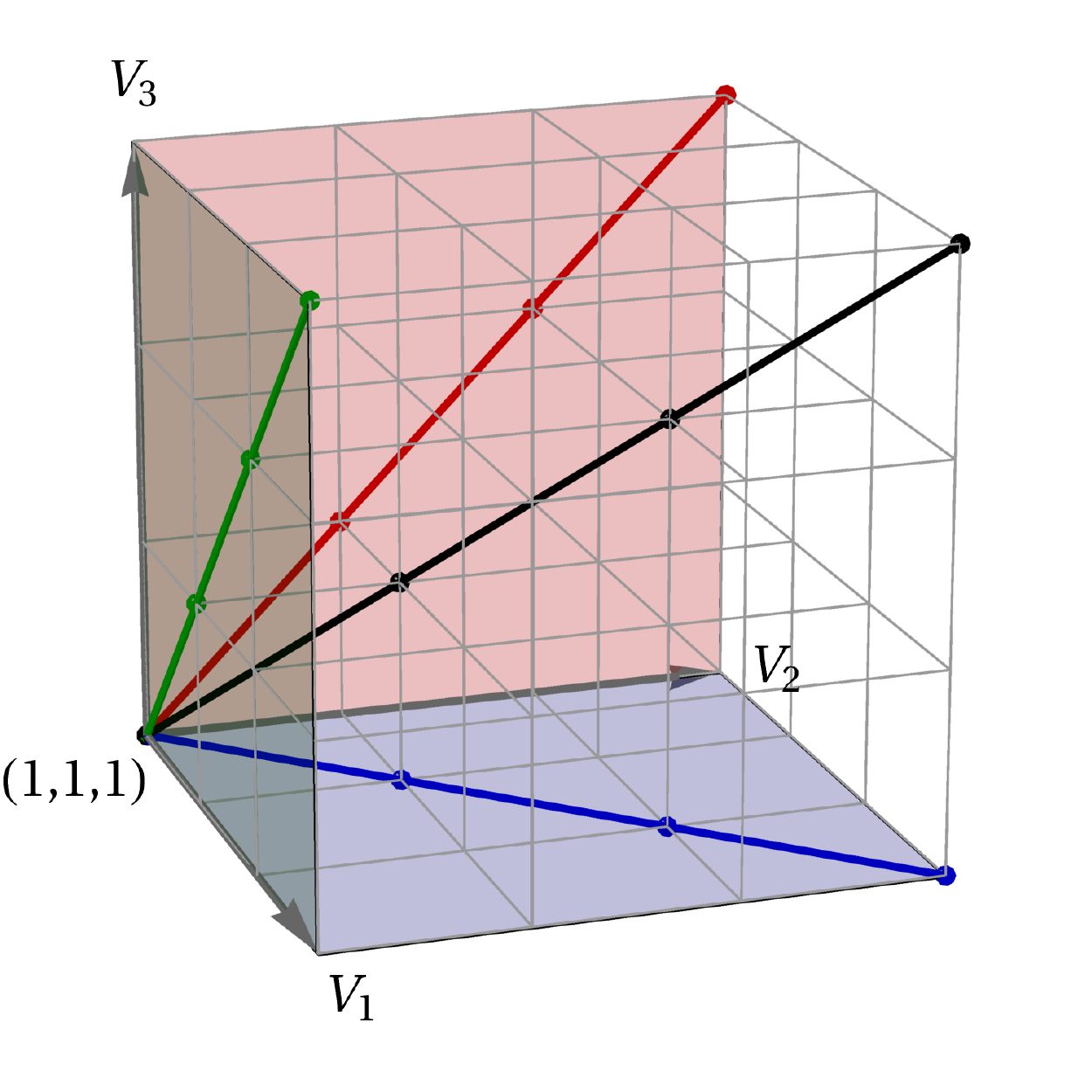}
\caption{The diagram representing the $S_3$-symmetric MEC state \eqref{v_s_3} for the tripartite system $D=(4,4,4)$, where only the lattice points on the four indicated intervals are nonzero.}
\label{444MES}
\end{figure}

Let us return to the $n$-partite hypercubic systems $D=(d,\dotsc,d)$.
We conjecture that one way to select an MES for such a system is to include in it an arbitrarily chosen corner of the hypercube and all the points that are on 2D side diagonals, 3D body diagonals, 4D hypercubic diagonals, etc., where all the diagonals contain the chosen corner.
Since the number of the $k$D hypercubic diagonals is the binomial coefficient $C(n,k)=n!(k!(n-k)!)^{-1}$ for $2\le k\le n$, and the number of all such diagonals is $\sum_{k=2}^n C(n,k) =2^n-n-1$, we find that the total number of points to be included is
\begin{align}
P(n,d)=1+(2^n-n-1)(d-1).
\end{align}

We note that for any typical state in the hypercubic system of $n$ $d$-qudits, the total number of lattice points in its diagram is almost surely $d^n$, and by Conjecture \ref{conjecture_typical_states_are_MESs} any such state is an MES with the probability $1$.
We therefore compare the numbers $P(n,d)$ and $d^n$ to see how much shorter such MESs are than typical states.
There is no difference between these two kinds of states for $d=1$ since $P(n,1)=1$.
The difference for qubits ($d=2$) becomes relatively smaller as $n$ increases because $2^{-n}P(n,2)=1-2^{-n}n$.
For $d\ge 3$, however, these MESs are significantly shorter for large $n$ since $d^{-n}P(n,d)=O(2^n d^{1-n})\to 0$, $n\to\infty$.

The above states are only one example of how to proceed with constructing MESs in hypercubic systems which are much shorter than typical states.

We choose the corner $e_{1,1}\otimes \dotsb \otimes e_{n,1}$ and arrive at the following conjecture.

\begin{conjecture}
  For the $n$-partite system with the dimensions $D=(d,\dotsc,d)$, where $d\ge 2$, the $S_n$-symmetric state
\begin{align}
    v&=e_{1,1}\otimes\dotsb\otimes  e_{n,1} +\sum_{l=2}^d \sum_{k=2}^n \sum_{\sigma\in C(S,k)} e_{1,1}\otimes\dotsb\otimes e_{\sigma_1-1,1}\otimes e_{\sigma_1,l}\otimes e_{\sigma_1+1,1}\otimes\dotsb \nn \\
    &\dotsb\otimes e_{\sigma_k-1,1}\otimes e_{\sigma_k,l}\otimes e_{\sigma_k+1,1} \otimes\dotsb\otimes e_{n,1},
    \label{v_combinations}
\end{align}
  is an MES and it has the principal invariant
  \begin{align}
  N_D= \max{(0,d^n-nd^2+n-1)}, \ d\ge 2.
  \end{align}
Here $C(S,k)$ is a set of all $k$-combinations of the set $S=\{1,\dotsc,n\}$ and  $\sigma=(\sigma_1,\dotsc,\sigma_k)$ is a subset of $k$ distinct elements of $S$. 
The number of elements of $C(S,k)$ equals the binomial coefficient $C(n,k)$.
  \label{conjecture_hypercubic_systems_n}
\end{conjecture}

We check that for $n=3$ we have
\begin{align}
S=\{1,2,3\}, \ C(S,2)=\{(1,2),(1,3),(2,3)\}, \ C(S,3)=\{(1,2,3)\}
\end{align}
and \eqref{v_combinations} becomes \eqref{v_s_3}.
Similarly, for $n=4$ we have
\begin{align}
&S=\{1,2,3,4\}, \ C(S,2)=\{(1,2),(1,3),(1,4),(2,3),(2,4),(3,4)\}, \nn \\
&C(S,3)=\{(1,2,3),(1,2,4),(1,3,4),(2,3,4)\}, \ C(S,4)=\{(1,2,3,4)\}  
\end{align}
and \eqref{v_combinations} becomes
\begin{align}
  v =\sum_{l=1}^d e_{1,l}\otimes e_{2,l}\otimes e_{3,l}\otimes e_{4,l} +\sum_{l=2}^d \bigl( e_{1,l}\otimes e_{2,l}\otimes e_{3,1}\otimes e_{4,1} + e_{1,l}\otimes e_{2,1}\otimes e_{3,l}\otimes e_{4,1} \nn \\
  + e_{1,l}\otimes e_{2,1}\otimes e_{3,1}\otimes e_{4,l} + e_{1,1}\otimes e_{2,l}\otimes e_{3,l}\otimes e_{4,1} + e_{1,1}\otimes e_{2,l}\otimes e_{3,1}\otimes e_{4,l} + e_{1,1}\otimes e_{2,1}\otimes e_{3,l}\otimes e_{4,l} \nn \\ 
  + e_{1,l}\otimes e_{2,l}\otimes e_{3,l}\otimes e_{4,1} + e_{1,l}\otimes e_{2,l}\otimes e_{3,1}\otimes e_{4,l} + e_{1,l}\otimes e_{2,1}\otimes e_{3,l}\otimes e_{4,l} + e_{1,1}\otimes e_{2,l}\otimes e_{3,l}\otimes e_{4,l} \bigr).
  \label{v_z_4}
\end{align}
In the corresponding $4$-dimensional diagram, the first term in \eqref{v_z_4} is the $4$D hypercubic diagonal, the first six terms in the second sum are the $2$D body diagonals, and the remaining four terms in the second sum are the $3$D side diagonals.

We prove Conjecture \ref{conjecture_hypercubic_systems_n} for $n=3$ as Theorem \ref{theorem_hypercubic_systems_n_3} in Appendix \ref{section_a_theorem_for_3_partite_hypercubic_systems}.

\subsection{Hypercuboid Systems}
\label{section_hypercuboid_systems}

Other cases for which we have examples of maximally entangled states are $n$-partite hypercuboid systems $D=(d_1,\dotsc,d_n)$ for which not all the dimensions of the subsystems $d_1,\dotsc,d_n$ are equal.
Since these systems do not have the symmetry of the hypercubic systems, this leads to new features in their entanglement classes.
We now provide a representative collection of examples by way of a set of tables of the most entangled classes for the systems that we have studied.
While the results cannot be summarized as easily as for the hypercubic systems in Section \ref{section_hypercubic_systems}, there are still some interesting systematic features to explore, i.e., taxonomy.
One thing to note is that, unlike in many of the hypercubic systems, in most hypercuboid examples the principal invariant $N_D$ is not necessarily the only invariant which is non-zero for states in the MEC.

\subsubsection{3-partite examples}
\label{3_partite_examples}

All results for the systems with the dimensions $D=(d_1,d_2,d_3)$ are symmetric with respect to permutations in $D$, and so it is sufficient to consider only the cases with $d_1\le d_2\le d_3$.
In all cases, among the $4$ invariants in the set $\mathcal{N}_D$, the principal invariant $N_D$ is the only invariant which is non-zero for all states in the MECs.

When MESs of a system have nonzero invariant values, this implies that the space available within the system is too limited to allow all possible would-be-entanglements.
This tension between the system volume and entanglement can be relaxed by increasing the size of one or more subsystem.
For example, in the tripartite case, increasing $d_3$ beyond $d_1 d_2$ relaxes the tension whereupon all invariant values vanish for states in the MEC.
Examples of this phenomenon will be presented below.

We have studied further properties of the systems with the dimensions $D=(d_1,d_2,d_3)$ for $2\le d_1\le 3$ and various $d_2$ and $d_3$, and found that the values of $N_D$ follow simple patterns, as described below.

For $D=(2,d_2,d_3)$, we find $N_{(2,d_2,d_3)}=\delta_{d_2,d_3}\max{(0,d_2-3)}$, where $\delta_{d_2,d_3}$ is the Kronecker delta.

For $D=(3,d_2,d_3)$, the results are somewhat more complicated; see Tables \ref{table_3_d2_d3} and  \ref{table_polynomials}.
%
%
\begin{table}[htpb]
\begin{center}
\newcolumntype{C}{>{\centering\arraybackslash}p{1.4em}}
\[\begin{array}{c@{\hspace{1.em}}*{16}{C}}
\toprule
& 1 & 2 & 3 & 4 & 5 & 6 & 7 & 8 & 9 & 10 & 11 & 12 & 13 & 14 & 15 & 16 \\[3pt]
 1 & 0 & 0 & 0 & 0 & 0 & 0 & 0 & 0 & 0 & 0 & 0 & 0 & 0 & 0 & 0 & 0 \\ 
 2 & 0 & 0 & 0 & 0 & 0 & 0 & 0 & 0 & 0 & 0 & 0 & 0 & 0 & 0 & 0 & 0 \\ 
 3 & 0 & 0 & 2 & 4 & 4 & 2 & 0 & 0 & 0 & 0 & 0 & 0 & 0 & 0 & 0 & 0 \\ 
 4 & 0 & 0 & 4 & 9 & 12 & 13 & 12 & 9 & 4 & 0 & 0 & 0 & 0 & 0 & 0 & 0 \\
 5 & 0 & 0 & 4 & 12 & 18 & 22 & 24 & 24 & 22 & 18 & 12 & 4 & 0 & 0 & 0 & 0 \\ 
 6 & 0 & 0 & 2 & 13 & 22 & 29 & 34 & 37 & 38 & 37 & 34 & 29 & 22 & 13 & 2 & 0 \\ 
 7 & 0 & 0 & 0 & 12 & 24 & 34 & 42 & 48 & 52 & 54 & 54 & 52 & 48 & 42 & 34 & 24 \\ 
 8 & 0 & 0 & 0 & 9 & 24 & 37 & 48 & 57 & 64 & 69 & 72 & 73 & 72 & 69 & 64 & 57 \\ 
 9 & 0 & 0 & 0 & 4 & 22 & 38 & 52 & 64 & 74 & 82 & 88 & 92 & 94 & 94 & 92 & 88 \\ 
10 & 0 & 0 & 0 & 0 & 18 & 37 & 54 & 69 & 82 & 93 & 102 & 109 & 114 & 117 & 118 & 117 \\ 
11 & 0 & 0 & 0 & 0 & 12 & 34 & 54 & 72 & 88 & 102 & 114 & 124 & 132 & 138 & 142 & 144 \\ 
12 & 0 & 0 & 0 & 0 & 4 & 29 & 52 & 73 & 92 & 109 & 124 & 137 & 148 & 157 & 164 & 169 \\ 
13 & 0 & 0 & 0 & 0 & 0 & 22 & 48 & 72 & 94 & 114 & 132 & 148 & 162 & 174 & 184 & 192 \\ 
14 & 0 & 0 & 0 & 0 & 0 & 13 & 42 & 69 & 94 & 117 & 138 & 157 & 174 & 189 & 202 & 213 \\ 
15 & 0 & 0 & 0 & 0 & 0 & 2 & 34 & 64 & 92 & 118 & 142 & 164 & 184 & 202 & 218 & 232 \\ 
16 & 0 & 0 & 0 & 0 & 0 & 0 & 24 & 57 & 88 & 117 & 144 & 169 & 192 & 213 & 232 & 249 \\ 
\bottomrule
\end{array}\]
\end{center}
\caption{The principal entanglement invariant $N_D$ as a function of $d_2$ and $d_3$ for the MECs for $D=(3,d_2,d_3)$.
For each of these cases, the full set of invariants is $\mathcal{N}_D=(0,0,0,N_D)$.
The sequence $N_{(3,d_2,d_3)}$ is given by the polynomials $N_{(3,d_2,d_3)}$ for $d_{3,\textrm{min}}\le d_3\le d_{3,\textrm{max}}$, which are listed in Table \ref{table_polynomials} for $3\le d_2\le 16$.}
\label{table_3_d2_d3}
\end{table}
%
%
Table \ref{table_3_d2_d3} is symmetric with respect to its main diagonal, $N_{(3,d_2,d_3)}=N_{(3,d_3,d_2)}$, and  the diagonal entries satisfy $N_{(3,d_3,d_3)}=\max{(0,d_3^2-7)}$.
We find $N_{(3,d_2,d_3)}=0$ for $1\le d_2\le 2$ and any $d_3$.
In the row corresponding to a fixed value of $d_2\ge 4$, the invariant $N_{(3,d_2,d_3)}$ is non-zero only for $\ceil{\frac{1}{3}d_2}+1\le d_3\le 3d_2-3$, where $\ceil{x}$ is the least integer greater than or equal to $x$.
For the same row, the invariant $N_{(3,d_2,d_3)}$ reaches the maximum either at $d_3=\frac{3}{2}d_2$ (when $d_2$ is even) or at $d_3\in\{\frac{1}{2}(3d_2-1),\frac{1}{2}(3d_2+1)\}$ (when $d_2$ is odd), and non-zero values of $N_{(3,d_2,d_3)}$ in this row are symmetric with respect to the these maximum values.
The corresponding patterns hold for columns of the table as a result of the symmetry of the table with respect to its main diagonal.

\begin{table}[htpb]
\begin{center}
\addtolength{\tabcolsep}{2pt}    
\begin{tabular}{ccc}
\toprule
$d_2$ & $N_{(3,d_2,d_3)}$ & $d_{3,\textrm{min}}\le d_3\le d_{3,\textrm{max}}$ \\
\midrule
$3$ & $-d_3^2+9d_3-16$ & $3\le d_3\le 6$ \\
$4$ & $-d_3^2+12d_3-23$ & $3\le d_3\le 9$ \\
$5$ & $-d_3^2+15d_3-32$ & $3\le d_3\le 12$ \\
$6$ & $-d_3^2+18d_3-43$ & $3\le d_3\le 15$ \\
$7$ & $-d_3^2+21d_3-56$ & $4\le d_3\le 17$ \\
$8$ & $-d_3^2+24d_3-71$ & $4\le d_3\le 20$ \\
$9$ & $-d_3^2+27d_3-88$ & $4\le d_3\le 23$ \\
$10$ & $-d_3^2+30d_3-107$ & $5\le d_3\le 25$ \\
$11$ & $-d_3^2+33d_3-128$ & $5\le d_3\le 28$ \\
$12$ & $-d_3^2+36d_3-151$ & $5\le d_3\le 31$ \\
$13$ & $-d_3^2+39d_3-176$ & $6\le d_3\le 33$ \\
$14$ & $-d_3^2+42d_3-203$ & $6\le d_3\le 36$ \\
$15$ & $-d_3^2+45d_3-232$ & $6\le d_3\le 39$ \\
$16$ & $-d_3^2+48d_3-263$ & $7\le d_3\le 41$ \\
\bottomrule
\end{tabular}
\addtolength{\tabcolsep}{-2pt}    
\end{center}
\caption{Polynomials $N_{(3,d_2,d_3)}$ and the ranges of their applicabilities $d_{3,\textrm{min}}\le d_3\le d_{3,\textrm{max}}$ for all the data given in Table \ref{table_3_d2_d3}.}
\label{table_polynomials}
\end{table}

\subsubsection{4-partite examples}
\label{4_partite_examples}
 
Now moving to $4$-partite systems, we encounter computational challenges that restrict the range of dimensions $D=(d_1,d_2,d_3,d_4)$ that can be explored even with the limited goal of studying only MECs.  
All results are symmetric with respect to permutations of the dimensions in $D$, and so it is sufficient to consider only the cases with $d_1\le d_2\le d_3\le d_4$.
In contrast to $3$-partite systems, we find that among the $19$ invariants in the set $\mathcal{N}_D$ of all $4$-partite invariants, the principal invariant $N_D$ is not necessarily the only invariant which is non-zero for states in the MEC.
We present the results for the principal invariant $N_D$ for a small selection of systems with  $D\in\{(2,2,2,d_4),(2,2,3,d_4),(2,2,4,d_4),(2,3,3,d_4)\}$ in Tables \ref{table_2_2_2_d4}, \ref{table_2_2_3_d4}, \ref{table_2_2_4_d4}, \ref{table_2_3_3_d4}.
Generalizing the results for the $3$-partite examples found in Section \ref{3_partite_examples}, we conjecture that for all the above $4$-partite cases, $N_D$ is non-zero only for $2\le d_4\le d_1 d_2 d_3 -2$ and that $N_D$ reaches the maximum either at $d_4=\frac{1}{2}d_1 d_2 d_3$ (when $d_1 d_2 d_3$ is even) or at $d_4\in\{\frac{1}{2}(d_1 d_2 d_3-1),\frac{1}{2}(d_1 d_2 d_3+1)\}$ (when $d_1 d_2 d_3$ is odd. 

\begin{table}[htpb]
\begin{center}
\begin{tabular}{cccccccc}
\toprule
$d_4$ & $1$ & $2$ & $3$ & $4$ & $5$ & $6$ & $7$ \\
\midrule
$N_{(2,2,2,d_4)}$ & $0$ & $3$ & $6$ & $7$ & $6$ & $3$ & $0$ \\
\bottomrule
\end{tabular}
\end{center}
\caption{The principal entanglement invariants $N_D$ for the MECs for $D=(2,2,2,d_4)$.
The sequence $N_D$ is given by the polynomial $N_{(2,2,2,d_4)}=-d_4^2+8d_4-9$ for $2\le d_4\le 6$.}
\label{table_2_2_2_d4}
\end{table}%

\begin{table}[htpb]
\begin{center}
\begin{tabular}{ccccccccccccccccc}
\toprule
$d_4$ & $1$ & $2$ & $3$ & $4$ & $5$ & $6$ & $7$ & $8$ & $9$ & $10$ & $11$ \\
\midrule
$N_{(2,2,3,d_4)}$ & $0$ & $6$ & $13$ & $18$ & $21$ & $22$ & $21$ & $18$ & $13$ & $6$ & $0$ \\
\bottomrule
\end{tabular}
\end{center}
\caption{The principal entanglement invariants $N_D$ for the MECs for $D=(2,2,3,d_4)$.
The sequence $N_D$ is given by the polynomial $N_{(2,2,3,d_4)}=-d_4^2+12d_4-14$ for $2\le d_4\le 10$.}
\label{table_2_2_3_d4}
\end{table}%

\begin{table}[htpb]
\begin{center}
\begin{tabular}{cccccccccccccccccc}
\toprule
$d_4$ & $1$ & $2$ & $3$ & $4$ & $5$ & $6$ & $7$ & $8$ & $9$ & $10$ & $11$ & $12$ & $13$ & $14$ & $15$ \\
\midrule
$N_{(2,2,4,d_4)}$ & $0$ & $7$ & $18$ & $27$ & $34$ & $39$ & $42$ & $43$ & $42$ & $39$ & $34$ & $27$ & $18$ & $7$ & $0$ \\
\bottomrule
\end{tabular}
\end{center}
\caption{The principal entanglement invariants $N_D$ for the MECs for $D=(2,2,4,d_4)$.
The sequence $N_D$ is given by the polynomial $N_{(2,2,4,d_4)}=-d_4^2+16d_4-21$ for $2\le d_4\le 14$.}
\label{table_2_2_4_d4}
\end{table}%

\begin{table}[htpb]
\begin{center}
\begin{tabular}{cccccccccccccccccc}
\toprule
$d_4$ & $1$ & $2$ & $3$ & $4$ & $5$ & $6$ & $7$ & $8$ & $9$ & $10$ & $11$ & $12$ & $13$ & $14$ & $15$ & $16$ & $17$ \\
\midrule
$N_{(2,3,3,d_4)}$ & $0$ & $13$ & $26$ & $37$ & $46$ & $53$ & $58$ & $61$ & $62$ & $61$ & $58$ & $53$ & $46$ & $37$ & $26$ & $13$ & $0$ \\
\bottomrule
\end{tabular}
\end{center}
\caption{The principal entanglement invariants $N_D$ for the MECs for $D=(2,3,3,d_4)$.
The sequence $N_D$ is given by the polynomial $N_{(2,3,3,d_4)}=-d_4^2+18d_4-19$ for $2\le d_4\le 16$.}
\label{table_2_3_3_d4}
\end{table}%

\subsubsection{5-partite examples}
\label{5_partite_examples}

Computations for $5$-partite systems are even more challenging. 
Among the $167$ invariants in the set $\mathcal{N}_D$ of all $5$-partite invariants, the principal invariant $N_D$ is not necessarily the only invariant which is non-zero for states in the MEC.
We present the results for the principal invariant $N_D$ only for two examples in Tables \ref{table_2_2_2_2_d5} and \ref{table_2_2_2_3_d5}.
A complete set of invariants $\mathcal{N}_D$ for the MEC for $D=(2,2,2,2,d_5)$ with the largest value of $N_D$, for example, is
\begin{align}
\mathcal{N}_{(2,2,2,2,8)} =(0^{18}, 5^1, 0^2, 5^1, 0^1, 5^2, 0^{87}, 43^1, 0^1, 43^2, 0^1, 43^3, 0^{25}, 25^2, 0^1, 25^2, 52^1, 0^{16}),
\end{align}
where the notation $a^m$ means that the invariant value $a$ is repeated $m$ times, and the last non-zero entry is the value of the principal invariant, i.e., $N_{(2,2,2,2,8)}=52$.
Generalizing the results for the $3$-partite  and $4$-partite examples found in Sections \ref{3_partite_examples} and \ref{4_partite_examples}, we conjecture that for all the above $5$-partite cases, $N_D$ is non-zero only for $1\le d_5\le d_1 d_2 d_3 d_4 -1$ and that $N_D$ reaches the maximum either at $d_5=\frac{1}{2}d_1 d_2 d_3 d_4$ (when $d_1 d_2 d_3 d_4$ is even) or at $d_5\in\{\frac{1}{2}(d_1 d_2 d_3 d_4-1),\frac{1}{2}(d_1 d_2 d_3 d_4+1)\}$ (when $d_1 d_2 d_3 d_4$ is odd). 


\begin{table}[htpb]
\begin{center}
\begin{tabular}{ccccccccccccccccc}
\toprule
$d_5$ & $1$ & $2$ & $3$ & $4$ & $5$ & $6$ & $7$ & $8$ & $9$ & $10$ & $11$ & $12$ & $13$ & $14$ & $15$ & $16$ \\
\midrule
$N_{(2,2,2,2,d_5)}$ & $3$ & $16$ & $27$ & $36$ & $43$ & $48$ & $51$ & $52$ & $51$ & $48$ & $43$ & $36$ & $27$ & $16$ & $3$ & $0$ \\
\bottomrule
\end{tabular}
\end{center}
\caption{The principal entanglement invariants $N_D$ for the MECs for $D=(2,2,2,2,d_5)$.
The sequence $N_D$ is given by the polynomial $N_{(2,2,2,2,d_5)}=-d_5^2+16d_5-12$ for $1\le d_5\le 15$.}
\label{table_2_2_2_2_d5}
\end{table}%

\begin{table}[htpb]
\begin{center}
\begin{tabular}{ccccccccccccccccccccccccc}
\toprule
$d_5$ & $1$ & $2$ & $3$ & $4$ & $5$ & $6$ & $7$ & $8$ & $9$ & $10$ & $11$ & $12$ & $13$ & $14$ & $15$ & $16$ & $17$ & $18$ & $19$ & $20$ & $21$ & $22$ & $23$ & $24$ \\
\midrule
$N_{(2,2,2,3,d_5)}$ & $6$ & $27$ & $46$ & $63$ & $78$ & $91$ & $102$ & $111$ & $118$ & $123$ & $126$ & $127$ & $126$ & $123$ & $118$ & $111$ & $102$ & $91$ & $78$ & $63$ & $46$ & $27$ & $6$ & $0$ \\
\bottomrule
\end{tabular}
\end{center}
\caption{The principal entanglement invariants $N_D$ for the MECs for $D=(2,2,2,3,d_5)$.
The sequence $N_D$ is given by the polynomial $N_{(2,2,2,3,d_5)}=-d_5^2+24d_5-17$ for $1\le d_5\le 23$.}
\label{table_2_2_2_3_d5}
\end{table}%
Higher dimensional $5$-partite examples become increasingly more difficult to generate.

\subsection{General Systems}
\label{section_general_systems}
 
The results in Sections \ref{section_hypercubic_systems} and \ref{section_hypercuboid_systems} lead to the following conjecture.

\begin{conjecture}
The principal entanglement invariant $N_D$ for any $n$-partite system with the dimensions $D=(d_1,\dotsc,d_n)$ is
\begin{align}
    N_{(2,d_2,d_3)} = N_{(2,d_3,d_2)} = N_{(d_2,2,d_3)} = N_{(d_2,d_3,2)} = N_{(d_3,2,d_2)} = N_{(d_3,d_2,2)} = \delta_{d_2,d_3}\max{(0,d_2-3)}
    \label{N_conjecture_2_d2_d3}
\end{align}
for $n=3$, $\min{(d_1,d_2,d_3)}=2$, where $\delta_{d_2,d_3}$ is the Kronecker delta, and
\begin{align}
    N_{(d_1,\dotsc,d_n)} =\max{\biggl(0,\ \prod_{i=1}^n d_i -\sum_{i=1}^n d_i^2 +n-1\biggr)}
    \label{N_conjecture}
\end{align}
for either $n=3$, $\min{(d_1,d_2,d_3)}\ge 3$, or $n\ge 4$.
\label{conjecture_N}
\end{conjecture}

We note that the quantity $N_{(d_1,\dotsc,d_n)}$ is a symmetric polynomial of the variables $d_1,\dotsc,d_n$ and prove some properties of $N_{(d_1,\dotsc,d_n)}$ in Appendix \ref{section_properties_of_the_principal_entanglement_invariants}.

According to Definition \ref{definition_typical_states} and Conjecture \ref{conjecture_typical_states_are_MESs}, if we randomly select values of coordinates $\{v_{j_1,\dotsc,j_n}\}$ uniformly from any finite region of $\C^{d_1\dotsb d_n}$ of dimension $d_1\dotsb d_n$, then the resulting state $v$ in \eqref{v} will almost surely be an MES for an $n$-partite system with the dimensions $D=(d_1,\dotsc,d_n)$.
Such a state has the maximum length, $L_D(v)=d_1\dotsb d_n$.

We can similarly take a number $l$, where $0\le l\le d_1\dotsb d_n$, and choose a set of $l$ coordinates among $\{v_{j_1,\dotsc,j_n}\}$.
There are $C(d_1\dotsb d_n,l)=(d_1\dotsb d_n)!(l!(d_1\dotsb d_n-l)!)^{-1}$ ways to make such a choice.
We set these chosen coordinates to be non-zero and randomly select their values uniformly among non-zero numbers from any finite region of $\C^{d_1\dotsb d_n}$ of dimension $d_1\dotsb d_n$.
The resulting state $v^{(l)}$ is a typical state with the length $L_D(v^{(l)})=l$.
If we exclude sets of values of measure zero, it does not matter what values from $\C^{d_1\dotsb d_n}$ are chosen for the non-zero coefficients of $v^{(l)}$, and each such state has the same set of invariants $\mathcal{N}_D(v^{(l)})$.
It follows that we need to generate only one random state $v^{(l)}$ for each $l$, which results in the total number of required states equal to $\sum_{l=0}^{d_1\dotsb d_n}C(d_1\dotsb d_n,l)=2^{d_1\dotsb d_n}$.

We can use the above direct method only for a limited set of systems since the number $2^{d_1\dotsb d_n}$ grows rapidly with the dimensions $d_1,\dotsc,d_n$.
For $D=(2,2,2)$, for example, there are $2^8=256$ possibilities, and the detailed information about distribution of these states among all entanglement classes according to their lengths is given in Figure \ref{figure_state_counts_by_length_and_class_222}.

 \begin{figure}[htpb]
    \centering
    \includegraphics{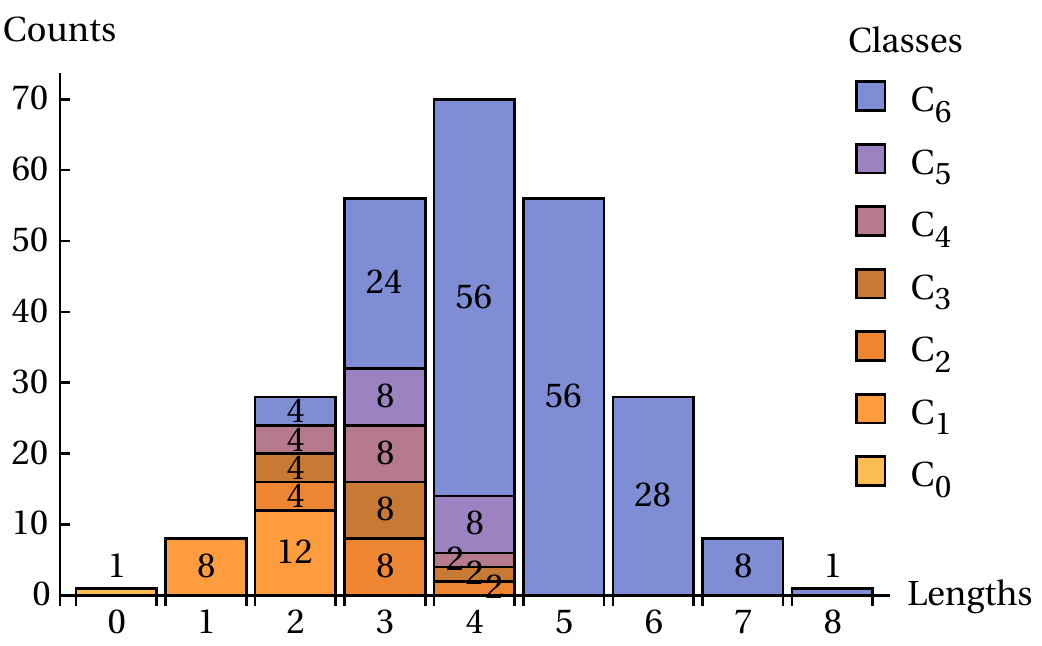}
    \caption{The distribution of all types of typical states for the system $D=(2,2,2)$ by their lengths and classes.
    The total number of types of states is $2^8=256$.
    The number of states with the length $L_D=l$ equals the binomial coefficient $C(8,l)=8!(l!(8-l)!)^{-1}$, but their distribution among the entanglement classes $C_0,\dotsc,C_6$ is less trivial.
    The distribution shifts from less entangled classes to more entangled classes as $L_D$ increases.
    In particular, the least entangled class $C_1$ requires $L_D\le 2$ and the most entangled class $C_6$ requires $L_D\ge 2$.
    The number of states of a given length that belong to a certain class is determined by permutations.
    For example, the $4$ states in the MEC $C_6$ of the shortest length $L_D=2$ are $v^{(2)}_k=\sum_{(j_1,j_2,j_3)\in J_k}v_{j_1,j_2,j_3}e_{1,j_1}\otimes e_{2,j_2}\otimes e_{3,j_3}$, where $1\le k\le 4$, $J_1=\{(1,1,1),(2,2,2)\}$, $J_2=\{(1,2,1),(2,1,2)\}$, $J_3=\{(2,1,1),(1,2,2)\}$, $J_4=\{(2,2,1),(1,1,2)\}$ and all $v_{j_1,j_2,j_3}$ specified by $\{J_k\}_{k=1}^4$ are arbitrary non-zero complex numbers.
    }
    \label{figure_state_counts_by_length_and_class_222}
\end{figure}

In the remainder of this subsection we look at the shortest lengths of states in MECs.
We list the results for several $3$-partite and $4$-partite systems in Tables \ref{table_shortest_length_3_partite} and \ref{table_shortest_length_4_partite}.
We give the exact values of $L_D$ for $3$-partite systems and only the upper bounds of $L_D$ for $4$-partite systems in view of computational challenges for the latter systems.

\begin{table}[htpb]
  \centering
  \begin{tabular}{ccc}
    \toprule
    $D$ & $N_D$ & $L_D$ \\
    \midrule
    $(2,2,2)$ &  $0$ &  $2$ \\
    $(2,2,3)$ &  $0$ &  $4$ \\
    $(2,2,4)$ &  $0$ &  $4$ \\
    $(2,2,5)$ &  $0$ &  $4$ \\
    $(2,3,3)$ &  $0$ &  $4$ \\
    $(2,3,4)$ &  $0$ &  $6$ \\
    $(2,3,5)$ &  $0$ &  $6$ \\
    $(2,4,4)$ &  $1$ &  $6$ \\
    $(2,4,5)$ &  $0$ &  $8$ \\
    $(2,5,5)$ &  $2$ &  $8$ \\
    $(3,3,3)$ &  $2$ &  $7$ \\
    $(3,3,4)$ &  $4$ &  $8$ \\
    $(3,3,5)$ &  $4$ &  $9$ \\
    $(3,4,4)$ &  $9$ &  $9$ \\
    $(3,4,5)$ & $12$ & $10$ \\
    $(3,5,5)$ & $18$ & $11$ \\
    $(4,4,4)$ & $18$ & $10$ \\
    $(4,4,5)$ & $25$ & $11$ \\
    $(4,5,5)$ & $36$ & $12$ \\
    $(5,5,5)$ & $52$ & $13$ \\
    \bottomrule
  \end{tabular}
  \caption{The principal invariants $N_D$ and the shortest lengths $L_D$ of states in MECs for various $D$ with $n=3$.}
\label{table_shortest_length_3_partite}
\end{table}

\begin{table}[htpb]
  \centering
  \addtolength{\tabcolsep}{4pt}    
  \begin{tabular}{ccl}
    \toprule
    $D$ & $N_D$ & $L_D$ \\
    \midrule
    $(2,2,2,2)$ &  $3$ & $6$ 
    \\
    $(2,2,2,3)$ &  $6$ & $8$ 
    \\
    $(2,2,2,4)$ &  $7$ & $\le 8$ 
    \\
    $(2,2,3,3)$ & $13$ & $\le 9$ 
    \\
    $(2,2,3,4)$ & $18$ & $\le 10$ 
    \\
    $(2,2,4,4)$ & $27$ & $\le 12$ 
    \\
    $(2,3,3,3)$ & $26$ & $\le 13$ \\
    $(2,3,3,4)$ & $37$ & $\le 14$ \\
    $(2,3,4,4)$ & $54$ & $\le 15$ \\
    $(2,4,4,4)$ & $79$ & $\le 17$ \\
    $(3,3,3,3)$ & $48$ & $\le 17$ \\
    $(3,3,3,4)$ & $68$ & $\le 21$ \\
    $(3,3,4,4)$ & $97$ & $\le 26$ \\
    $(3,4,4,4)$ & $138$ & $\le 31$ \\
    $(4,4,4,4)$ & $195$ & $\le 44$ \\
    \bottomrule
  \end{tabular}
  \addtolength{\tabcolsep}{-2pt}    
  \caption{The principal invariant $N_D$ and the shortest lengths $L_D$ of states in MECs for various $D$ with $n=4$.
  In this table, all possibilities for states $v^{(l)}$ with $L_D(v^{(l)})=l$ for the first two cases were considered.
  Since the number of required states $2^{d_1 d_2 d_3 d_4}$ grows rapidly with the dimensions $d_1,d_2,d_3,d_4$, only the Monte Carlo methods could be used for all other cases and such computations can only provide upper bounds.
  The next four examples were studied with $10^5$ samples each, and the remaining cases were studied with fewer samples.}
\label{table_shortest_length_4_partite}
\end{table}

\pagebreak

The results in Table \ref{table_shortest_length_3_partite} lead to the following conjecture.
\begin{conjecture}
The shortest length of states in an MEC for $n=3$ is $L_{(2,d_2,d_2)}=2d_2-2$, $L_{(2,d_2,d_3)}=2d_2$ for $d_2<d_3$ and
\begin{align}
    L_{(d_1,d_2,d_3)} =\sum_{i=1}^3 d_i-2
    \label{length_conjecture}
\end{align}
for $\min{(d_1,d_2,d_3)}\ge 3$.
\end{conjecture}

It would be interesting to generalize this conjecture to $n\ge 4$, but at present we do not have sufficient numerical data to carry this out.

\section{A Random Walk Through the Space of Entanglement Classes}
\label{section_a_random_walk_through_the_space_of_entanglement_classes}

Hard scattering of short wavelength particles off a quantum system can drive that system toward decoherence, while more gentle scattering of long wavelength particles can drive the system toward coherence \cite{Jacob:2021}. If we consider scattering of particles one by one, then it generates something similar to a random walk (RW), where the high frequency case  is closely related to Brownian motion. Consideration of classical random walks are sufficient for our purposes here, although quantum random walks \cite{KempeQRW,VenegasQW} used for search algorithms by a quantum computer could be employed in a more sophisticated approach.
 
We think of our RWs taking place in the Born-Markov (BM) approximation, which makes the following three assumptions.
First, we assume that the environment (E) is much larger than the system (S) and the coupling between them is weak, so that the density matrix can be written in the factorized form $\rho(t)=\rho_\text{S}(t)\otimes \rho_\text{E}$.
The large environment assumption means the interactions with S have little effect on E, so that $\rho_\text{E} \approx \text{const}$.
Second, self-correlations within E created by interactions with S decay rapidly, i.e., there is no long-term memory in E over time scales which are comparable to the time scale for S to change.
Finally, we  assume that the interaction Hamiltonian can be written in the form $H_\text{int}=H_{\text{int},\text{S}}\otimes H_{\text{int},\text{E}}$.
These assumptions lead to the BM master equation for the evolution of $\rho(t)$.
We do not need to analyze this equation directly as we mention it here only to justify our Monte Carlo calculations.
In our RWs we follow only pure states in $\text{S}$, so if $\rho_\text{S}(t)$ were to evolve into a mixed state,
then  we in effect follow each component, as these travel along the paths shown in the following figures and lead to the same conclusions as to how states reach MECs.
 
Let us consider random walks of states in our system through the space of its entanglement classes to investigate how distant unentangled states are from the MEC.
A tripartite case $D=(d_1,d_2,d_3)$ is sufficiently rich in various features of these random walks and so we restrict our attention to this case.
A general state \eqref{v} has the form
\begin{align}
v=\sum_{j_1=1}^{d_1} \sum_{j_2=1}^{d_2} \sum_{j_3=1}^{d_3} v_{j_1,j_2,j_3} e_{1,j_1}\otimes e_{2,j_2}\otimes e_{3,j_3}.
\label{v_2_2_2}
\end{align}
We begin at the vacuum state $v^{(0)}=0$ and generate a sequence $(v^{(0)},v^{(1)}\dotsc,v^{(d_1 d_2 d_3)})$ of states in a random walk such that the state $v^{(l+1)}$ is obtained from $v^{(l)}$ by setting an arbitrary zero coefficient $v_{j_1,j_2,j_3}$ in $v^{(l)}$ to an arbitrary non-zero value.
The general form of the state $v^{(l)}$ in the sequence is
\begin{align}
v^{(l)}=\sum_{(j_1,j_2,j_3)\in D_3^{(l)}} v_{j_1,j_2,j_3} e_{1,j_1}\otimes e_{2,j_2}\otimes e_{3,j_3},\label{v_l}
\end{align}
where
\begin{align}
v_{j_1,j_2,j_3}\not=0, \ (j_1,j_2,j_3)\in D_3^{(l)}
\end{align}
and $D_3^{(l)}$ is a set of any $l$ distinct elements of the set
\begin{align}
D_3=\{1,\dotsc,d_1\}\times\{1,\dotsc,d_2\}\times\{1,\dotsc,d_3\}.
\end{align}

There are interesting questions about an ensemble of such random walks.
For example, although it is clear that the last element in the sequence, $v^{(d_1 d_2 d_3)}$, is an MES, we might reach an MES much earlier, and so we ask how many steps it takes to get from $v^{(0)}$ to the first MES in the sequence.

Corresponding to a sequence of states $(v^{(0)},v^{(1)},\dotsc,v^{(d_1 d_2 d_3)})$, there is a sequence of entanglement classes $(C^{(0)},C^{(1)},\dotsc,C^{(d_1 d_2 d_3)})$ such that $v^{(l)}\in C^{(l)}$.
Although states in their sequence are all distinct by construction, some classes in their sequence can coincide.
For example, for the states
\begin{align}
&v^{(1)}=v_{j_1,j_2,j_3} e_{1,j_1}\otimes e_{2,j_2}\otimes e_{3,j_3}, \\
&v^{(2)}=v_{j_1,j_2,j_3} e_{1,j_1}\otimes e_{2,j_2}\otimes e_{3,j_3} +v_{k_1,j_2,j_3} e_{1,k_1}\otimes e_{2,i_2}\otimes e_{3,i_3}
\end{align}
with $j_1\not=k_1$, we have $C^{(1)}=C^{(2)}$ since $v^{(1)}$ and $v^{(2)}$ are both unentangled, as can be seen by the basis redefinition
\begin{align}
&v^{(2)}=e'_{1,j_1}\otimes e_{2,j_2}\otimes e_{3,j_3}, \\
&e'_{1,j_1}=v_{j_1,j_2,j_3} e_{1,j_1}+v_{k_1,j_2,j_3} e_{1,k_1}.
\end{align}
Also, if $C^{(l)}$ and $C^{(l+1)}$ are distinct, then the class $C^{(l)}$ cannot appear in the sequence after $C^{(l+1)}$.
Equivalently, for the ordering of the entanglement classes generated by our random walk procedure (for which, by the way, the canonical choice does not exist), the sequence of entanglement classes generated in an RW is a monotonically non-decreasing sequence.
Consequently, if $C^{(l)}$ is the MEC, then any $C^{(m)}$ with $m>l$ is also the MEC, and so $C^{(m)}=C^{(l)}$ for all $m>l$.
This means that from the point of view of sequences of classes, RWs terminate at the MECs.
However, before the sequence reaches the MEC, it visits other classes, and not necessarily all of them.
This leads us to another question regarding an ensemble of such random walks, namely how often do they visit various entanglement classes on their way from $C^{(0)}$ to the MEC?

For specific results of RWs, consider first the system $D=(2,2,2)$.
Generating 10,000 RW examples, we find the pathways shown in Figure \ref{figure_traffic_pattern_tripartite}, where the individual paths connecting classes (which we call sub-paths) are for single steps in the RW.
Not all entanglement classes are connected by a single step, and there are many possible paths from the vacuum class $C_0$ to the MEC $C_6$.
Some sub-paths are used more frequently than others, giving an indication of how likely it is for an RW to reach a certain class.
The sub-paths that loop back to the same class are those where turning on a new coefficient does not change the class.
The probabilities of RWs taking various sub-paths are shown in the figure.
If there is no path connecting any two given classes, then such a transition is not possible and the associated probability is zero.

We note that all the information available in Figure \ref{figure_traffic_pattern_tripartite} is also given by the stochastic matrix $P$ in Eq. \eqref{P_Q_2_2_2} of Appendix \ref{section_markov_chains}.
Such matrices appear when transitions between the entanglement classes are viewed as transitions between events in Markov chains.
In our specific case, for each $D=(d_1,d_2,d_3)$ system we have an absorbing Markov chain with the MEC being its only absorbing class.
See Appendix \ref{section_markov_chains} for more details.

\begin{figure}[htpb]
\includegraphics[scale=0.4]{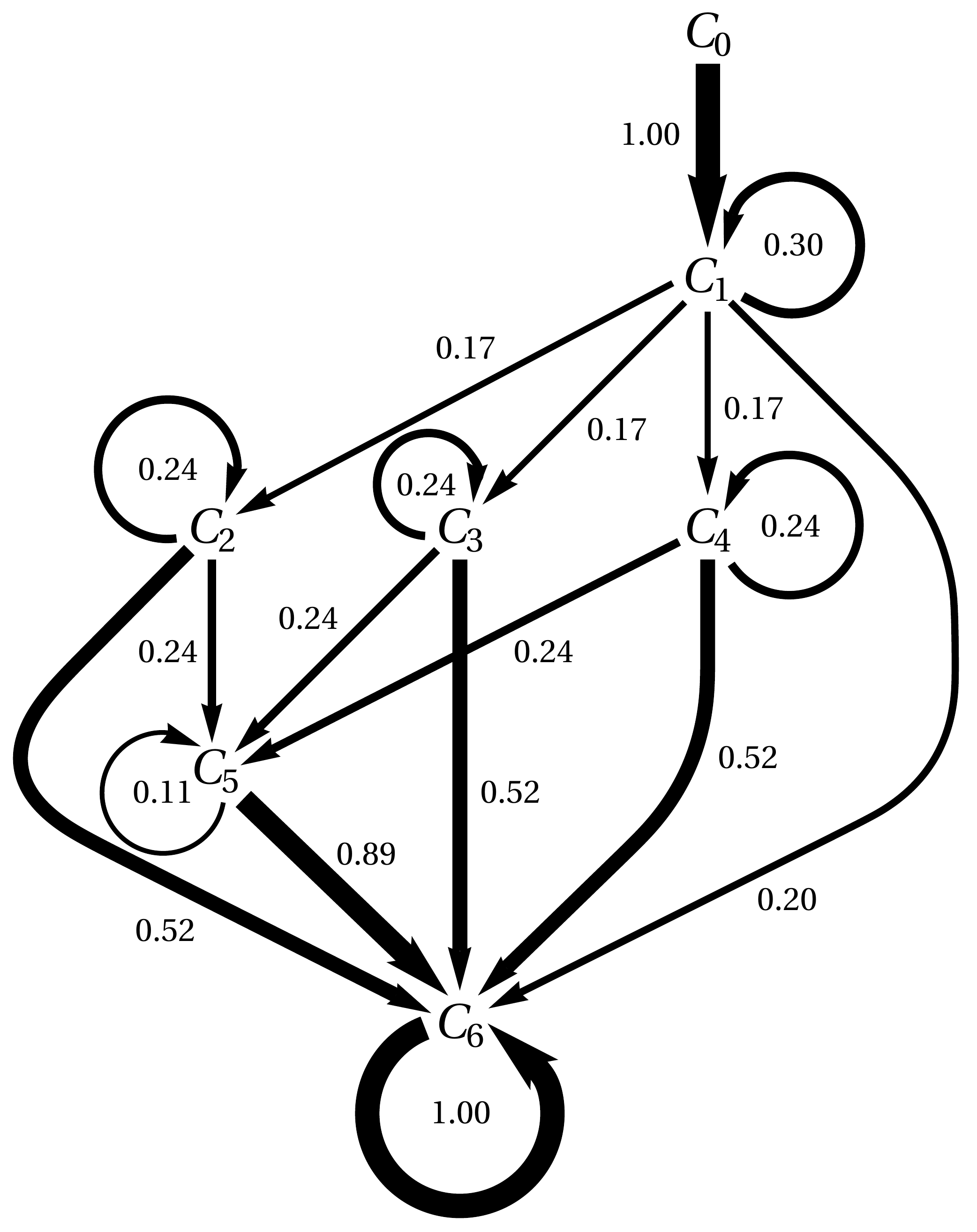}%
\caption{The traffic pattern for random walks from the vacuum class $C_0$ to the MEC $C_6$ for the tripartite system $D=(2,2,2)$.
The numbers next to arrows mean the probabilities of all random walks to take the paths indicated by the arrows.
The same notation is also used for the paths that do not change classes and are indicated as self-loops.
Once the self-loops are included, the sum of the probabilities for arrows leaving any class equals one (subject to small variations due to rounding of numbers).
If there is no path connecting any two given classes, then such a transition is not possible and the associated probability is zero.
All the information available in this figure is also given by the stochastic matrix $P$ in \eqref{P_Q_2_2_2}.
Such a matrix gives transition probabilities between events in a Markov chain that corresponds to random walks between entanglement classes. 
}
\label{figure_traffic_pattern_tripartite}
\end{figure}

\subsection{Random Walk Examples}
\label{section_random_walk_examples}

Figures similar to Figure \ref{figure_traffic_pattern_tripartite} for larger systems are rather cluttered, and so in these cases we instead look at heat maps for their evolution.
Figure \ref{figure_forward_heat_map} shows such a heat map for the system $D=(3,3,3)$ and includes flow patterns.
The heat map displays the number of RWs where a class is reached at a certain step of the process of turning coefficients on, with colors changing from black (for less than four RWs) to white (for all $2\times 10^5$ RWs).
An arrow shows the average of all taken paths from a given step to the next one, with its length indicating the change in class.
Some details (similar to those in Figure \ref{figure_traffic_pattern_tripartite}) are omitted for clarity, allowing the overall pattern to be easily seen.
Again, certain classes are more or less likely to be hit in any RW.
We can extract the lengths of the RWs to reach the MEC, and do this for many other systems as well. 
While RWs that travel from the least to the most entangled class can be thought of as heating the system by scattering high energy particles, we can also think of cooling the system.
This could occur via scattering of low-frequency photons similar to laser cooling.
We can model this situation by simply turning off coefficients one by one, which reverses the heating process.
However, there is more than one way to do this, e.g., we can start with all coefficients on and turn them off one by one, but we could also start with a data set of initial states that are the set of end points of the heating RWs.
These have only a non-maximum number ($<d_1\dotsb d_n$) of coefficients turned on.
However, a return path, i.e., a cooling path, can be different from the heating path that generated its initial conditions, although the number of steps in both directions has to coincide.
Figure \ref{figure_return_heat_map} shows such a return heat map based on the end points reached in Figure \ref{figure_forward_heat_map}.

\begin{figure}[htpb]%
    \centering
    \subfloat[\label{figure_forward_heat_map}Forward heat map.\hspace*{-3em}]
    {\includegraphics{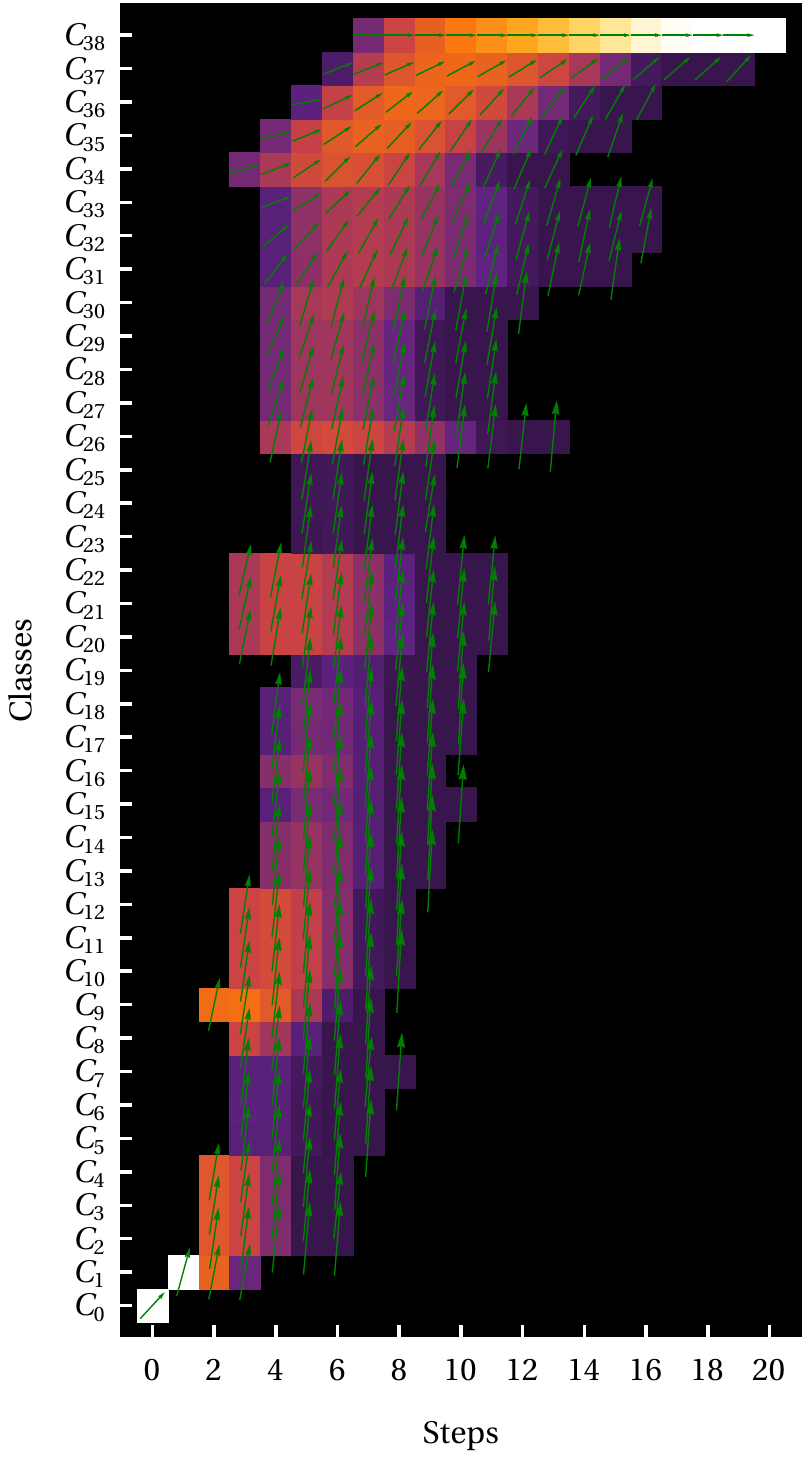}}%
    \hfill
    \subfloat[\label{figure_return_heat_map}Return heat map.\hspace*{3.5em}]
    {\includegraphics{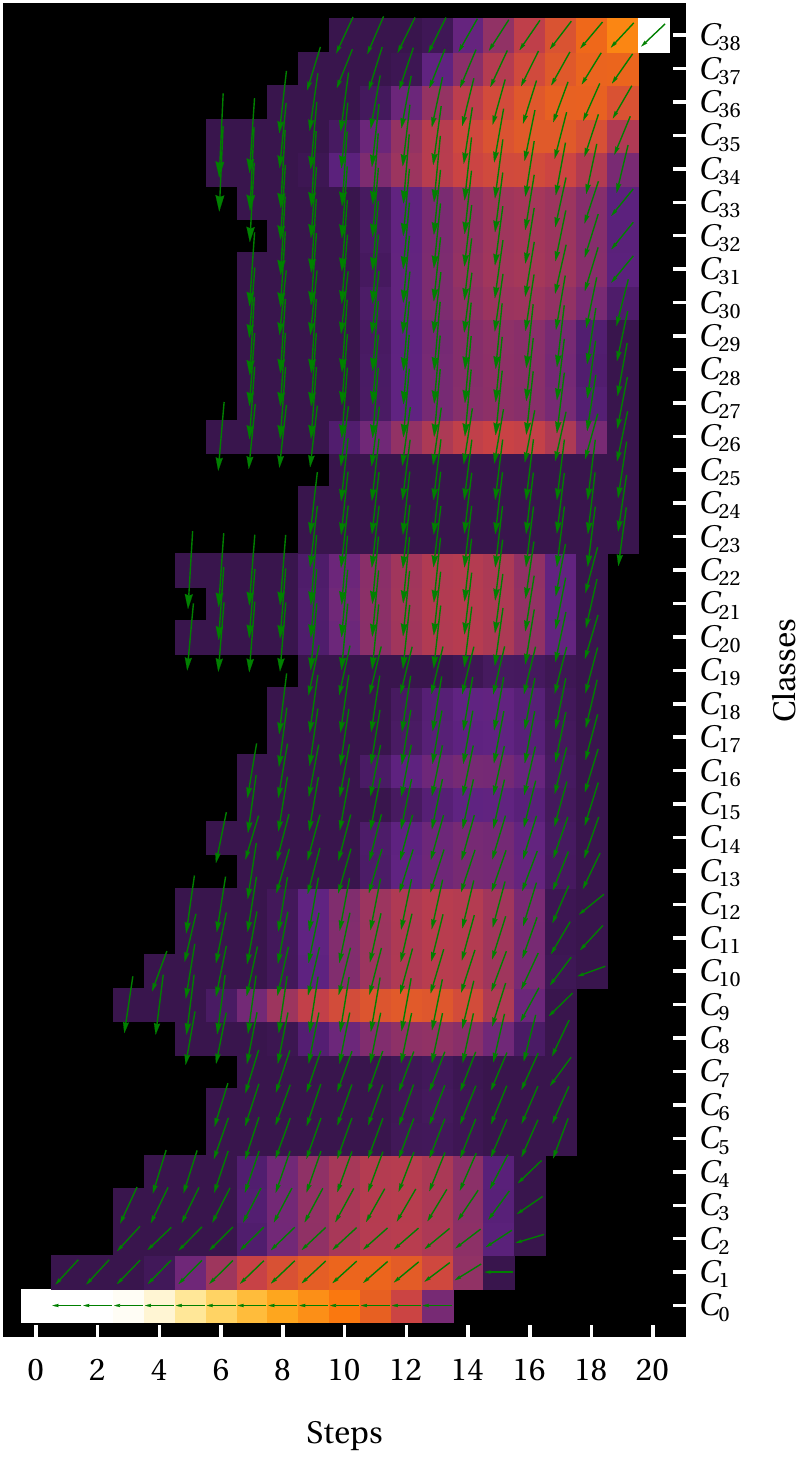}}%
    \caption{Heat maps for RWs between the vacuum class $C_0$ and the MEC $C_{(3,3,3)}=C_{38}$ for the tripartite system $D=(3,3,3)$ with $2\times10^{5}$ samples and four hits minimum.}%
\end{figure}

\begin{table}[htpb]
  \centering
  \addtolength{\tabcolsep}{4pt}    
  \begin{tabular}{crrccc}
    \toprule
    $D$ & $\smallabs{\mathcal{C}_D}$ & $\mathcal{N}_D$ & $Q_0$ & $\tilde{Q}_0$ & $L_D$ \\
    \midrule
    $(2,2,2)$ &$7$       & $(0,0,0,0)$    & $3.6$  &$3$  & $2$ \\
    $(3,3,3)$ &$39$      & $(0,0,0,2)$    & $11.0$ &$7$  & $7$ \\
    $(4,4,4)$ &$\ge 166$ & $(0,0,0,18)$   & $13.9$ &$11$ & $10$ \\
    $(5,5,5)$ &$\ge 407$ & $(0,0,0,52)$   & $18.3$ &$15$ & $13$ \\
    $(6,6,6)$ &$\ge 759$ & $(0,0,0,110)$  & $22.9$ &$19$ & $16$ \\
     $(7,7,7)$ &$ ?$ & $(0,0,0,198)$   & $?$ &$23$ & $19$ \\
    $(d,d,d)$ & ?        & $(0,0,0,N_{(d,d,d)})$ & $\approx(4d-2)$&$4d-5$ & $\approx(3d-2)$ \\
    \bottomrule
  \end{tabular}
  \caption{Several properties of the MECs of tripartite systems with the dimensions  $D=(d,d,d)$, including the numbers of entanglement classes $\smallabs{\mathcal{C}_D}$ (either as exact values or lower bounds), the full sets of invariants $\mathcal{N}_D$, the approximate average numbers of steps $Q_0$ from the vacuum class to the MECs, the numbers of steps $\tilde{Q}_0$ from the vacuum class to the MECs for the $S_3$-symmetric states \eqref{v_s_3}, and the minimum lengths $L_D$ of MESs.
  The value of the principal invariant for the MEC is $N_{(d,d,d)}=d^3-3d^2+2$ for $d\ge 3$.
  The question marks for several entries indicate that we do not have reliable ways to compute the required quantities.}
  \label{table_MEC_properties_ddd}
\end{table}

\begin{table}[htpb]
  \begin{center}
    \begin{tabular}{ccr}
      \toprule
      Class & $\mathcal{N}_{(2,2,2)}$ & Count \\
      \midrule
      $C_0$ & $(8,8,8,8)$ & $10,000$ \\
      $C_1$ & $(4,4,4,4)$ & $14,256$ \\
      $C_2$ & $(0,0,4,3)$ & $3,236$ \\
      $C_3$ & $(0,4,0,3)$ & $3,064$ \\
      $C_4$ & $(4,0,0,3)$ & $3,207$ \\
      $C_5$ & $(0,0,0,1)$ & $2,560$ \\
      $C_6$ & $(0,0,0,0)$ & $10,000$ \\
      \bottomrule
    \end{tabular}
  \end{center}
\label{table_222_rw_counts}
\caption{The number of classes which are visited in $10,000$ random walks from the unentangled class to the MEC for the tripartite system $D=(2,2,2)$.
The value of the principle invariant $N_{(2,2,2)}$ is the fourth number in the set $\mathcal{N}_{(2,2,2)}$.}
\end{table}%

 \begin{figure}[htpb]
    \centering
    \includegraphics{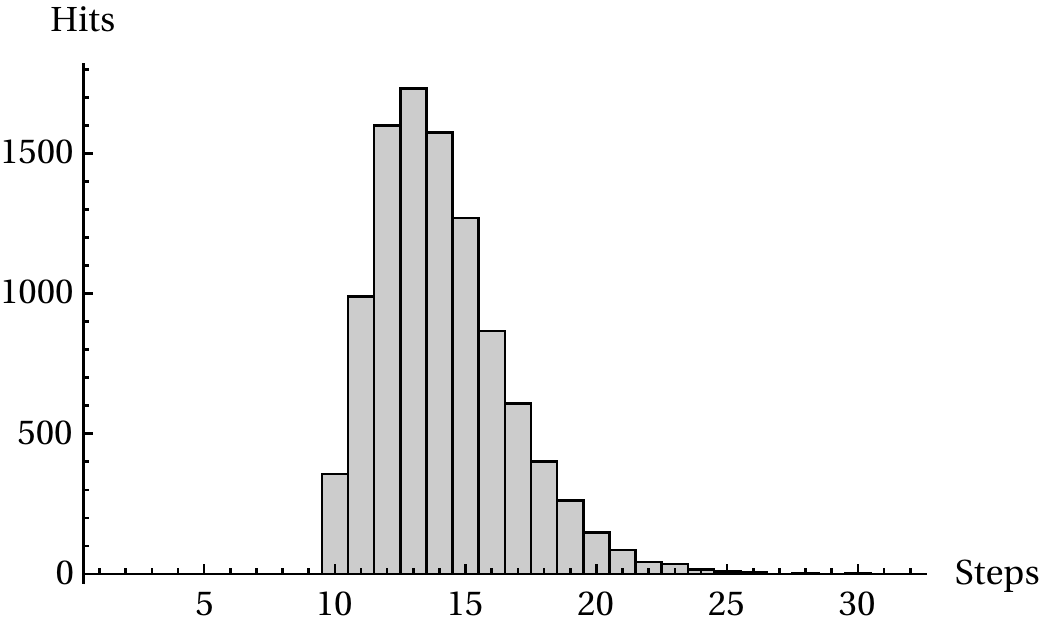}
    \caption{A histogram of the number of steps needed to reach the MEC for the system $D=(4,4,4)$ for $10,000$ random walks.
    The maximum is at 13 steps and the average number of steps is approximately $14$.
    The smallest number of steps is $10$, which agrees with Conjecture~\ref{conjecture_N} and Table \ref{table_MEC_properties_ddd}.}
    \label{figure_random_walks_histogram_444}
\end{figure}

We collected several properties of the MECs of tripartite systems with the dimensions $D=(d,d,d)$ in Table \ref{table_MEC_properties_ddd}.
We included the numbers of entanglement classes $\smallabs{\mathcal{C}_D}$ (either as exact values or lower bounds), the full sets of invariants $\mathcal{N}_D$, the approximate average numbers of steps $Q_0$ from the vacuum class to the MECs, the numbers of steps $\tilde{Q}_0$ from the vacuum class to the MECs for the $S_3$-symmetric states \eqref{v_s_3}, and the minimum lengths $L_D$ of MESs.
We were not able to reliably compute several of these quantities and we indicated them with the question marks in the table. 

We have found that the values of $Q_0$, $\tilde{Q}_0$, $L_D$ grow approximately linearly with $d$ for $D=(d,d,d)$.
Figure \ref{figure_random_walks_histogram_444} shows a histogram of the number of steps needed to reach the MEC for 10,000 RWs for the larger system $D=(4,4,4)$.
It exhibits the minimum of $10$ and the average of about $14$ steps in accordance with Table \ref{table_MEC_properties_ddd}.  
This can be compared to Figure \ref{444MES} which implies that the length of a symmetric state is $4d-3$ and contrasted with the length of the minimal RW. Hence the symmetric representation of an MEC as given by Figure \ref{444MES} is not atypical of the number of terms needed to reach the MEC via a RW, and this number grows far slower than the total number of possible nonzero coefficients for a state in the $D=(d,d,d)$ system. For an increasing system size this can be seen from a decreasing average fraction of the above number of non-zero coefficients to reach the MEC and the total number of coefficients. This is visualized in Figure \ref{figure_MEC_nonzero_coefficient_fraction_heat_maps} for the systems $D=(2,d_2,d_3)$ and $D=(3,d_2,d_3)$ with $1\leq d_2,d_3\leq 6$.  

\begin{figure}[htpb]
    \centering
    \subfloat[$D=(2,d_2,d_3)$.\hspace*{-20pt}]{\includegraphics{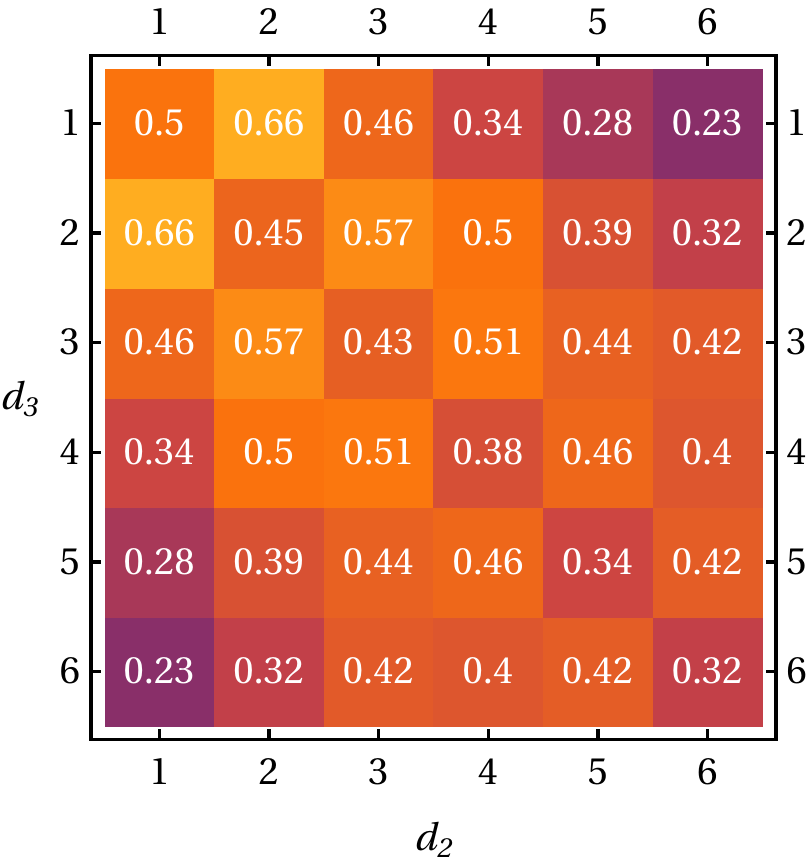}}%
    \hspace{-7.5pt}%
    \subfloat[$D=(3,d_2,d_3)$.\hspace*{20pt}]{\includegraphics{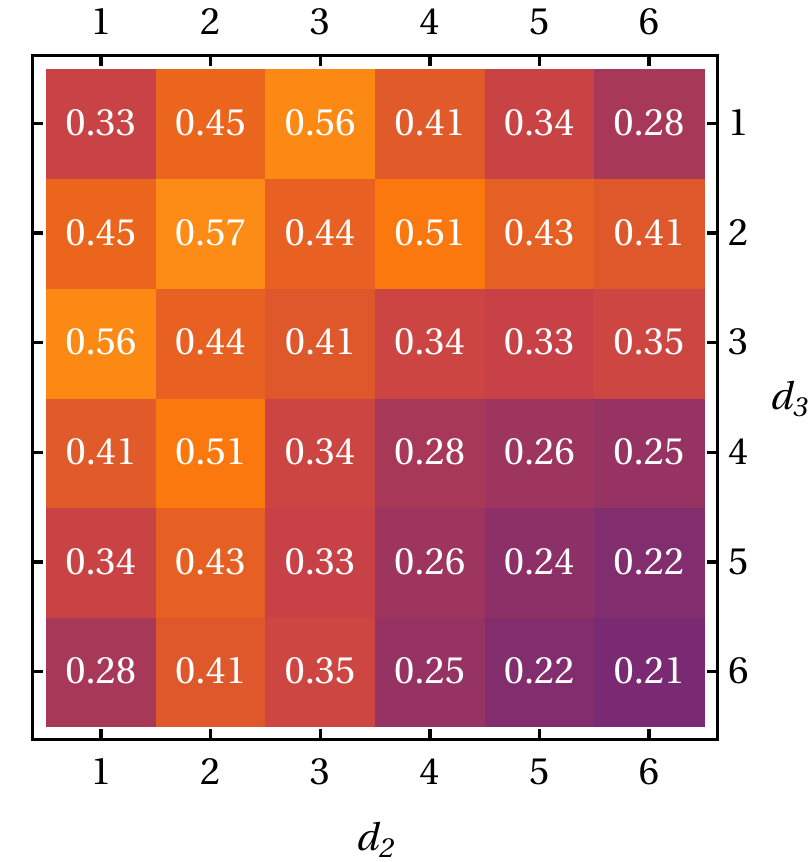}}%
    \caption{The average fraction of non-zero coefficients of the MEC reached in $1000$ random walks for the systems $D=(2,d_2,d_3)$ and $D=(3,d_2,d_3)$ with $1\leq d_2,d_3\leq 6$.}
    \label{figure_MEC_nonzero_coefficient_fraction_heat_maps}
\end{figure}

\section{Discussion and Conclusions}
\label{section_discussion_and_conclusions}

Quantum devices (including quantum computers) depend on maintaining quantum systems in specific quantum states (usually specific entangled states), and then being able to manipulate those states.
This is in general a daunting task as unwanted interactions of the system with its environment tend to cause decoherence of the desired states.
Decoherence effects can be fought by shielding and super cooling  the system, and through quantum error correction, but the difficulty grows with the system size.

The results given here can be used in estimates of decoherence, and specifically in estimating the step-wise distance  from unentangled to most entangled classes.
This quantity turns out to grow approximately linearly with system size, while the number of classes grows much faster.
We hope this type of information will be useful for design and operation of quantum computers and other quantum devices and for a general understanding of quantum behavior.
The method that we have described can be easily generalized to find the distance between any two classes of an arbitrary quantum system.

There are examples in particle physics where the decay of a particle leads naturally to a specific entanglement.
If the decay takes place in a good vacuum, then entanglement can be maintained until the decay products interact with the measurement apparatus.
Entanglement in particle scattering and decays has been of interest for some time  \cite{ALP01,Bertlmann06,Bramon07,Bernabeu11,Banerjee15,Peschanski16,Kharzeev17,Cervera-Lierta:2017tdt,CerveraLierta:2019ejf} with one focus being on testing Bell-type inequalities \cite{Benatti97,Bertlmann01,Bertlmann01bis}.
These studies have typically used or only needed a more limited classification of quantum states from what we have presented here.
Hence, we anticipate being able to exploit our more refined classification to explore more complicated decay and scattering processes.
Entanglement has also been applied to astrophysics and cosmology \cite{Maldacena15,Berera2020,Berera2021xqa}.

Multipartite quantum states called ``absolutely maximally entangled states,'' have been defined, which are multipartite maximally entangled with respect to any possible bipartition \cite{HC13}.
Entangled $n$–qudits states via orthogonal arrays and combinatorial designs  \cite{Goyeneche14,Seveso18,Goyeneche14} and other highly entangled states have also been studied by a number of other authors \cite{Arnaud13,Brown05,Facchi09,Vicente13,Cervera-Lierta:2018inm,Cervera-Lierta:2019}, but we caution the reader that in some cases maximally entangled means only with respect to bipartition, and so differs from our definition of MES.

Pointer states are robust in the sense that they are environmentally super selected states that correspond to physical quantities that can be read off a dial in an experimental apparatus \cite{Zurek:1982ii}.
As we can see from our random walks, the most entangled class is robust in that once obtained it resists change due to further environmental perturbation.
Hence, we can call the MEC a pointer class.
The MEC is a set of quasiclassical states that occupies a decoherence-free subspace (DFS) of the Hilbert space of the system and could potentially be used for quantum computation.
For a simple related example see the discussion of the spin-boson model in \cite{Schlosshauer} and references therein.
    
Importantly, we have shown that it is possible to generate short symmetric MESs in the MEC for a number of systems and to use them as a tool in exploring  random walks through class space and other details of the MEC.
The short $S_3$ symmetric states we construct require only $4d-5$ terms for the systems $D=(d,d,d)$, while the  number of terms in a typical random member of the MEC is $d^3$.
Due to their small size, we expect our compact states to find use in many other situations.
We have also suggested how to extend our construction to other systems. 

The importance of multipartite entanglement cannot be overemphasized \cite{Kuwahara:2021mkh}, and we expect it to be ever more important as the field of quantum information progresses.

\appendix

\section{A Theorem for the $\bm{S_3}$-symmetric states of 3-partite Hypercubic Systems}
\label{section_a_theorem_for_3_partite_hypercubic_systems}

\begin{theorem}
  For the 3-partite system with the dimensions $D=(d,d,d)$, where $d\ge 2$, the $S_3$-symmetric state
  \begin{align}
    v =\sum_{l=1}^d e_{1,l}\otimes e_{2,l}\otimes e_{3,l} +\sum_{l=2}^d (e_{1,1}\otimes e_{2,l}\otimes e_{3,l} +e_{1,l}\otimes e_{2,1}\otimes e_{3,l} +e_{1,l}\otimes e_{2,l}\otimes e_{3,1})
    \label{}
  \end{align}
  is an MES and it  has the algebraic invariants $\mathcal{N}_{(d,d,d)}=(0,0,0,N_{(d,d,d)})$, where the principal invariant is
  \begin{align}
  N_{(d,d,d)} = \max{(0,d^3-3d^2+2)}, \ d\ge 2.
  \end{align}
  \label{theorem_hypercubic_systems_n_3}
\end{theorem}
\begin{proof}
  Since the coordinates of the state $v$ are
  \begin{align}
    v_{i,j,k} =-3\delta_{i,1}\delta_{j,1}\delta_{k,1} +\delta_{i,j}\delta_{i,k}\delta_{j,k} +\delta_{i,1}\delta_{j,k} +\delta_{j,1}\delta_{i,k} +\delta_{k,1}\delta_{i,j}, \ 1\le i,j,k\le d,
    \label{}
  \end{align}
  it follows that the required equations in the definitions of $K_1$, $K_2$, $K_3$ \cite{Buniy:2010yh,Buniy:2010zp} are, respectively,
  \begin{align}
    &\sum_{i=1}^d v_{i,j,k} w_{i} =-3\delta_{j,1}\delta_{k,1}w_1 +\delta_{j,k}w_j +\delta_{j,k}w_1 +\delta_{j,1}w_k +\delta_{k,1}w_j =0, \ 1\le j,k\le d, \label{k_1_equation} \\
    &\sum_{j=1}^d v_{i,j,k} w_{j} =-3\delta_{i,1}\delta_{k,1}w_1 +\delta_{i,k}w_i +\delta_{i,k}w_1 +\delta_{i,1}w_k +\delta_{k,1}w_i =0, \ 1\le i,k\le d, \label{k_2_equation} \\
    &\sum_{k=1}^d v_{i,j,k} w_{j} =-3\delta_{i,1}\delta_{j,1}w_1 +\delta_{i,j}w_i +\delta_{i,j}w_1 +\delta_{i,1}w_j +\delta_{j,1}w_i =0, \ 1\le i,j\le d, \label{k_3_equation}
  \end{align}
  and $K_{1,2,3}$ is defined \cite{Buniy:2010yh,Buniy:2010zp} by the system of equations 
  \begin{align}
    &\sum_{j,k=1}^d v_{i,j,k} w_{i_0,j,k} =-3\delta_{i,1}w_{i_0,1,1} +\delta_{i,1}\sum_{l=1}^d w_{i_0,l,l} +w_{i_0,1,i} +w_{i_0,i,1} +w_{i_0,i,i} =0, \  1\le i,i_0\le d, \label{k_1_2_3_equation_1} \\
    &\sum_{i,k=1}^d v_{i,j,k} w_{i,j_0,k} =-3\delta_{j,1}w_{1,j_0,1} +\delta_{j,1}\sum_{l=1}^d w_{l,j_0,l} +w_{1,j_0,j} +w_{j,j_0,1} +w_{j,j_0,j} =0, \  1\le j,j_0\le d, \label{k_1_2_3_equation_2}  \\
    &\sum_{i,j=1}^d v_{i,j,k} w_{i,j,k_0} =-3\delta_{k,1}w_{1,1,k_0} +\delta_{k,1}\sum_{l=1}^d w_{l,l,k_0} +w_{1,k,k_0} +w_{k,1,k_0} +w_{k,k,k_0} =0, \  1\le k,k_0\le d. \label{k_1_2_3_equation_3} 
  \end{align}
  The algebraic invariants are $\mathcal{N}=(n_1,n_2,n_3,n_{1,2,3})$, where $n_i=\dim{K_i}$, $1\le i\le 3$ and $n_{1,2,3}=\dim{K_{1,2,3}}$.

  The space $K_1$ is the set of vectors $w\in V_1$ satisfying \eqref{k_1_equation}.
  For $j=k=1$ this equation becomes $w_1=0$, and for $j=1$, $k\ge 2$ it becomes $w_k=0$.
  All the remaining equations in \eqref{k_1_equation} are trivially satisfied by this zero solution $w_i=0$, $1\le i\le d$.
  This implies that $K_1$ contains only the zero vector and hence $n_1=0$.
  Similarly, the equations \eqref{k_2_equation} and \eqref{k_3_equation} also have only the zero solutions and we conclude that $n_2=0$ and $n_3=0$.

  We now calculate $n_{1,2,3}$.

  Let $d=2$.
  The only solution of the combined system of equations \eqref{k_1_2_3_equation_1}, \eqref{k_1_2_3_equation_2}, \eqref{k_1_2_3_equation_3} is $w_{i,j,k}=0$ for all $1\le i,j,k\le d$.
  This gives $n_{1,2,3}=0$.
  
  Let $d\ge 3$.
  Noting the special role played by the index $1$ and the $S_3$ symmetry of the state $v$, we have
  \begin{align}
  &w_{i,1,1}=w_{1,i,1}=w_{1,1,i}, \ 1\le i\le d, \label{w_S_3_symmetry_1} \\ 
  &w_{1,i,i}=w_{i,1,i}=w_{i,i,1}, \ 1\le i\le d. \label{w_S_3_symmetry_2} 
  \end{align}
  We split the equations \eqref{k_1_2_3_equation_1}, \eqref{k_1_2_3_equation_2}, \eqref{k_1_2_3_equation_3} into the following sets:
  \begin{align}
  &\sum_{l=1}^d w_{i_0,l,l}=0, \ 1\le i_0\le d, \label{k_1_2_3_equation_4} \\
  &\sum_{l=1}^d w_{l,j_0,l}=0, \ 1\le j_0\le d, \label{k_1_2_3_equation_5} \\
  &\sum_{l=1}^d w_{l,l,k_0}=0, \ 1\le k_0\le d, \label{k_1_2_3_equation_6} \\
  &w_{i_0,1,i} +w_{i_0,i,1} +w_{i_0,i,i} =0, \ 2\le i\le d, \ 1\le i_0\le d, \label{k_1_2_3_equation_7} \\
  &w_{1,j_0,j} +w_{j,j_0,1} +w_{j,j_0,j} =0, \ 2\le j\le d, \ 1\le j_0\le d, \label{k_1_2_3_equation_8}  \\
  &w_{1,k,k_0} +w_{k,1,k_0} +w_{k,k,k_0} =0, \ 2\le k\le d, \ 1\le k_0\le d. \label{k_1_2_3_equation_9} 
  \end{align}
  The solutions of \eqref{k_1_2_3_equation_7}, \eqref{k_1_2_3_equation_8}, \eqref{k_1_2_3_equation_9} are
  \begin{align}
  &w_{1,1,i}=-\tfrac{1}{2}w_{1,i,i}, \ 2\le i\le d, \label{k_1_2_3_solution} \\
  &w_{i_0,i,i}=-w_{i_0,1,i}-w_{i_0,i,1}, 2\le i,i_0\le d, \label{k_1_2_3_solution_7} \\
  &w_{j,j_0,j}=-w_{1,j_0,j}-w_{j,j_0,1}, 2\le j,j_0\le d, \label{k_1_2_3_solution_8} \\
  &w_{k,k,k_0}=-w_{1,k,k_0}-w_{k,1,k_0}, 2\le k,k_0\le d. \label{k_1_2_3_solution_9}
  \end{align}
  In view of \eqref{w_S_3_symmetry_1}, there are $3(d-1)$ solutions in \eqref{k_1_2_3_solution}.
  Taking $i=i_0$ in \eqref{k_1_2_3_solution_7}, $j=j_0$ in \eqref{k_1_2_3_solution_8}, and $k=k_0$ in \eqref{k_1_2_3_solution_9}, we have $d-1$ solutions $w_{i,i,i}=-2w_{1,i,i}$, $2\le i\le d$.
  The number of the remaining solutions in \eqref{k_1_2_3_solution_7}, \eqref{k_1_2_3_solution_8} and \eqref{k_1_2_3_solution_9} is $3(d-1)(d-2)$.
  Taking into account the $2(d-1)$ equations in \eqref{w_S_3_symmetry_2}, we find the total number of solutions
  \begin{align}
  3(d-1)+(d-1)+3(d-1)(d-2)+2(d-1)=3d(d-1),
  \end{align}
  which equals the number $3d(d-1)$ of equations in \eqref{k_1_2_3_equation_7}, \eqref{k_1_2_3_equation_8}, \eqref{k_1_2_3_equation_9}.
  
  In view of the symmetry \eqref{w_S_3_symmetry_2}, the equations \eqref{k_1_2_3_equation_4} for $i_0=1$, \eqref{k_1_2_3_equation_5} for $j_0=1$ and \eqref{k_1_2_3_equation_6} for $k_0=1$ reduce to only one equation,
  \begin{align}
  w_{1,1,1}+\sum_{l=2}^d w_{1,l,l}=0. \label{k_1_2_3_equation_10}
  \end{align}
  Using the solutions \eqref{k_1_2_3_solution}, \eqref{k_1_2_3_solution_7}, \eqref{k_1_2_3_solution_8} and \eqref{k_1_2_3_solution_9}, we write all the remaining equations in \eqref{k_1_2_3_equation_4}, \eqref{k_1_2_3_equation_5} and \eqref{k_1_2_3_equation_6} in the form
  \begin{align}
  &\tfrac{1}{2}w_{1,i_0,i_0}+\sum_{l=2}^d (w_{i_0,1,l}+w_{i_0,l,1})=0, \ 2\le i_0\le d, \label{k_1_2_3_equation_11} \\
  &\tfrac{1}{2}w_{j_0,1,j_0}+\sum_{l=2}^d (w_{1,j_0,l}+w_{l,j_0,1})=0, \ 2\le j_0\le d, \label{k_1_2_3_equation_12} \\
  &\tfrac{1}{2}w_{k_0,k_0,1}+\sum_{l=2}^d (w_{1,l,k_0}+w_{l,1,k_0})=0, \ 2\le k_0\le d. \label{k_1_2_3_equation_13}
  \end{align}
  We note that $3(d-1)$ equations \eqref{k_1_2_3_equation_11}, \eqref{k_1_2_3_equation_12} and \eqref{k_1_2_3_equation_13} contain $3(d-1)^2$ variables
  \begin{align}
  \{w_{1,k,l}\}_{2\le k,l\le d} \cup \{w_{k,1,l}\}_{2\le k,l\le d} \cup \{w_{k,l,1}\}_{2\le k,l\le d}. \label{w_variables}
  \end{align}
  Among these variables we choose $3(d-1)$ variables
  \begin{align}
  \{w_{1,2,3},w_{1,3,4},\dotsc,w_{1,d-1,d},w_{1,d,2}\} \cup \{w_{2,1,3},w_{3,1,4},\dotsc,w_{d-1,1,d},w_{d,1,2}\} \nn \\
  \cup \{w_{2,3,1},w_{3,4,1},\dotsc,w_{d-1,d,1},w_{d,2,1}\} \label{w_variables_1}
  \end{align}
  and write equations \eqref{k_1_2_3_equation_11}, \eqref{k_1_2_3_equation_12} and \eqref{k_1_2_3_equation_13} in the form $AW=B$.
  Here $W$ is a $3(d-1)\times 1$ matrix of the variables in \eqref{w_variables_1}, $B$ is a $3(d-1)\times 1$ matrix that depends on the variables in \eqref{w_variables} that are not in \eqref{w_variables_1}, and $A$ is the $3(d-1)\times 3(d-1)$ block matrix
  \begin{align}
  A=
  \begin{bmatrix}
  I & I & O \\
  O & I & I \\
  J & O & I
  \end{bmatrix}
  ,
  \end{align}
  where $O$ is the $(d-1)\times(d-1)$ zero matrix, $I$ is the $(d-1)\times(d-1)$ identity matrix and $J$ is the $(d-1)\times(d-1)$ cyclic permutation matrix
  \begin{align}
  J=
  \begin{bmatrix}
  0 & 1 & 0 & 0 & \dotsb & 0 & 0 \\
  0 & 0 & 1 & 0 & \dotsb & 0 & 0 \\
  \dotsb & \dotsb & \dotsb & \dotsb & \dotsb & \dotsb & \dotsb \\
  0 & 0 & 0 & 0 & \dotsb & 0 & 1 \\
  1 & 0 & 0 & 0 & \dotsb & 0 & 0 \\
  \end{bmatrix}
  .
  \end{align}
  (Note that we need to choose a specific ordering of elements of \eqref{w_variables_1} in $W$ to arrive at the above form of $A$.)
  
  We write the block matrix $A$ in the form of the block LU decomposition, $A=LDU$, where
  \begin{align}
  L=
  \begin{bmatrix}
  I & O & O \\
  O & I & O \\
  J & -J & I
  \end{bmatrix}
  , \ 
  D=
  \begin{bmatrix}
  I & O & O \\
  O & I & O \\
  O & O & I+J
  \end{bmatrix}
  , \ 
  U=
  \begin{bmatrix}
  I & I & O \\
  O & I & I \\
  O & O & I
  \end{bmatrix}
  .
  \end{align}
  The lower and upper block triangular matrices $L$ and $U$ are invertible and have the maximal rank $3(d-1)$, and
  \begin{align}
  \rank{D}=2\rank{I}+\rank{(I+J)}.
  \end{align}
  Using $\rank{I}=d-1$ and
  \begin{align}
  \rank{(I+J)}=
  \begin{cases}
  d-1, & d \textrm{\ is even}, \\ 
  d-2, & d \textrm{\ is odd},
  \end{cases}
  \end{align}
  we find
  \begin{align}
  \rank{A}=\rank{D}=
  \begin{cases}
  3(d-1), & d \textrm{\ is even}, \\ 
  3(d-1)-1, & d \textrm{\ is odd}.  
  \end{cases}
  \end{align}
  
  For an even $d$, the matrix $A$ is invertible and the equation $AW=B$ can be solved for $W$.
  We substitute the solution $W=A^{-1}B$ into the only remaining equation \eqref{k_1_2_3_equation_10} and solve the resulting equation for one of the remaining unknowns, for example, $w_{1,1,1}$.
  
  For an odd $d$, any given row of $A$ is linearly dependent on all the other rows there.
  To avoid this linear dependence (and the resulting non-invertibility of $A$), we remove any one equation from $AW=B$ and the corresponding variable from \eqref{w_variables_1}.
  The resulting equation $A'W'=B'$ is now solvable for $W'$ because $A'$ is invertible.
  We substitute the solution $W'=A^{\prime -1}B$ into the above removed equation and solve it for one of the variables in $\{w_{1,l,l}\}_{l=2}^d$.
  Finally, with the above solutions, we solve the only remaining equation \eqref{k_1_2_3_equation_10} for one of the remaining unknowns, for example, $w_{1,1,1}$. 
  
  As a result, we solved all equations in \eqref{k_1_2_3_equation_4}, \eqref{k_1_2_3_equation_5}, \eqref{k_1_2_3_equation_6}, \eqref{k_1_2_3_equation_7}, \eqref{k_1_2_3_equation_8}, \eqref{k_1_2_3_equation_9} for all even and odd $d$.
  We note that the difference between the even and odd cases does not effect the counting of the free parameters that the solutions depend on because in the even case these parameters include one extra parameter from the set $\{w_{1,l,l}\}_{l=2}^d$ in exchange for one extra parameter among \eqref{w_variables_1} in the odd case.
  To count the free parameters for both even and odd $d$, we have $d^3$ quantities $\{w_{i,j,k}\}_{i,j,k=1}^d$ satisfying $3d$ relations \eqref{k_1_2_3_equation_4}, \eqref{k_1_2_3_equation_5}, \eqref{k_1_2_3_equation_6} (among which there are $2$ identities resulting from taking $i_0=j_0=k_0=1$) and $3d(d-1)$ relations \eqref{k_1_2_3_equation_7}, \eqref{k_1_2_3_equation_8}, \eqref{k_1_2_3_equation_9}, which says that there are
  \begin{align}
  d^3-(3d-2)-3d(d-1)=d^3-3d^2+2
  \end{align}
  free parameters that the solutions depend on.
  This means that $n_{1,2,3}=\dim{K_{1,2,3}}=d^3-3d^2+2$. 
\end{proof}

\section{Properties of the Principal Entanglement Invariants}
\label{section_properties_of_the_principal_entanglement_invariants}

In this appendix we consider either $n=3$, $\min{(d_1,d_2,d_3)}\ge 3$ or $n\ge 4$, corresponding to \eqref{N_conjecture}.

Let $1\le i_1\le n$ and suppose that $N_{(\dotsc,d_{i_1}+\delta_{i_1},\dotsc)}>0$ for $\delta_{i_1}\in\{-1,0,1\}$, where from now on we use the short notation where we list only indices in fixed positions that may change and assume that all the remaining indices are identical in any subsequent appearances of $N_D$ in a given set of equations.
Using \eqref{N_conjecture} we find the recurrence relation
\begin{align}
    \sum_{\delta_{i_1}\in\{-1,1\}}N_{(\dotsc,d_{i_1}+\delta_{i_1},\dotsc)} =2N_{(\dotsc,d_{i_1},\dotsc)} -2.
    \label{recurrence_1}
\end{align}
Note that \eqref{recurrence_1} says that the second finite difference of $N_{(\dotsc,d_{i_1},\dotsc)}$ with respect to $d_{i_1}$ equals $-2$, which agrees with the dependence of $N_{(d_1,\dotsc,d_n)}$ on $d_i$ given in \eqref{N_dependence_on_d_i}.
If we know $N_{(\dotsc,d_{i_1}+\delta_{i_1},\dotsc)}>0$ for $\delta_{i_1}\in\{-1,0\}$ as initial conditions, then we can use \eqref{recurrence_1} to find recursively $N_{(\dotsc,d_{i_1}+j_1,\dotsc)}$ for all $j_1\ge 1$. 
We can similarly move in the direction of the decreasing values of $d_{i_1}$ and find recursively $N_{(\dotsc,d_{i_1}-1-j_1,\dotsc)}$ for all $j_1\ge 1$.
This allows us to fill out all entries in any row or column of the table of values of $N_D$ if we know the values of any two adjacent positive values in that row or column.

Let $1\le i_1<i_2\le n$ and suppose that we know $N_{(\dotsc,d_{i_1}+\delta_{i_1},\dotsc,d_{i_2}+\delta_{i_2},\dotsc)}>0$ for $\delta_{i_1},\delta_{i_2}\in\{-1,0\}$ as initial conditions.
We can use \eqref{recurrence_1} to find recursively
\begin{align}
N_{(\dotsc,d_{i_1}+j_1,\dotsc,d_{i_2}-1,\dotsc)}, \ N_{(\dotsc,d_{i_1}+j_1,\dotsc,d_{i_2},\dotsc)}, \  N_{(\dotsc,d_{i_1}-1,\dotsc,d_{i_2}+j_2,\dotsc)}, \ N_{(\dotsc,d_{i_1},\dotsc,d_{i_2}+j_2,\dotsc)}
\end{align}
for all $j_1\ge 1$ and $j_2\ge 1$.
Since now we have at least two known values of $N_D$ in each row (in the $d_{i_1}$ direction) and each column (in the $d_{i_2}$ direction), we can proceed recursively and fill entries for all values of $d_{i_1}$ and $d_{i_2}$.

More generally, let $1\le k\le n$ and $1\le i_1<\dotsb<i_k\le n$ and suppose that we know $2^k$ values  $N_{(\dotsc,d_{i_1}+\delta_{i_1},\dotsc,d_{i_k}+\delta_{i_k},\dotsc)}>0$ for $\delta_{i_1},\dotsc,\delta_{i_k}\in\{-1,0\}$ as initial conditions.
Similarly using \eqref{recurrence_1} we can find recursively the invariants $N_{(\dotsc,d_{i_1},\dotsc,d_{i_k},\dotsc)}$ for all values of $d_{i_1},\dotsc,d_{i_k}$.
The recurrence relation \eqref{recurrence_1} relates the values of $N_D$ on a line (a $1$-dimensional linear subset of $\Z^n$).
Using \eqref{recurrence_1} we can derive analogous recurrence relations that relate the values of $N_D$ on a $k$-dimensional linear subset of $\Z^n$:
\begin{align}
    \sum_{\delta_{i_1},\dotsc,\delta_{i_k}\in\{-1,1\}}N_{(\dotsc,d_{i_1}+\delta_{i_1},\dotsc,d_{i_k}+\delta_{i_k},\dotsc)} =2^k N_{(\dotsc,d_{i_1},\dotsc,d_{i_k},\dotsc)} -2^k k.
    \label{recurrence_k}
\end{align}
The recurrence relations \eqref{recurrence_k} may give a specific value of $N_D$ faster than repeated use of \eqref{recurrence_1}. 

As an example, let us use \eqref{recurrence_1} to compute several values of $N_{(3,d_2,d_3)}$ in Table \ref{table_3_d2_d3} using the initial conditions $N_{(3,3,3)}=2$, $N_{(3,3,4)}=4$, $N_{(3,4,3)}=4$, $N_{(3,4,4)}=9$:
\begin{align}
    N_{(3,3,5)} &=\max{(0,2N_{(3,3,4)}-N_{(3,3,3)}-2)} =\max{(0,2\times 4-2-2)} =4, \\
    N_{(3,3,6)} &=\max{(0,2N_{(3,3,5)}-N_{(3,3,4)}-2)} =\max{(0,2\times 4-4-2)} =2, \\
    N_{(3,3,7)} &=\max{(0,2N_{(3,3,6)}-N_{(3,3,5)}-2)} =\max{(0,2\times 2-4-2)} =0, \\
    N_{(3,4,5)} &=\max{(0,2N_{(3,4,4)}-N_{(3,4,3)}-2)} =\max{(0,2\times 9-4-2)} =12, \\
    N_{(3,4,6)} &=\max{(0,2N_{(3,4,5)}-N_{(3,4,4)}-2)} =\max{(0,2\times 12-9-2)} =13, \\
    N_{(3,4,7)} &=\max{(0,2N_{(3,4,6)}-N_{(3,4,5)}-2)} =\max{(0,2\times 13-12-2)} =12, \\
    N_{(3,4,8)} &=\max{(0,2N_{(3,4,7)}-N_{(3,4,6)}-2)} =\max{(0,2\times 12-13-2)} =9, \\
    N_{(3,4,9)} &=\max{(0,2N_{(3,4,8)}-N_{(3,4,7)}-2)} =\max{(0,2\times 9-12-2)} =4, \\
    N_{(3,4,10)} &=\max{(0,2N_{(3,4,9)}-N_{(3,4,8)}-2)} =\max{(0,2\times 4-9-2)} =0.
\end{align}
Using \eqref{recurrence_k} for $k=2$, we find
\begin{align}
    N_{(3,5,5)} =\max{(0,4N_{(3,4,4)} -N_{(3,3,3)} -N_{(3,3,5)} -N_{(3,5,3)} -8)} =\max{(0,4\times 9 -2 -4 -4 -8)} =18,
\end{align}
where we used the permutation symmetry $N_{(3,5,3)}=N_{(3,3,5)}$.

Let $1\le i\le n$ and consider how $N_{(d_1,\dotsc,d_n)}$ depends on $d_i$ for fixed values of $d_1,\dotsc,d_{i-1},d_{i+1},\dotsc,d_n$.
Assuming $N_{(d_1,\dotsc,d_n)}>0$, we have a quadratic polynomial in $d_i$,
\begin{align}
    N_{(d_1,\dotsc,d_n)} =-d_i^2 +d_i P_i -S_i +n-1,
    \label{N_dependence_on_d_i}
\end{align}
where
\begin{align}
    P_i=\prod_{\begin{subarray}{c} 1\le j\le n \\ j\not=i \end{subarray}}d_j, \  S_i=\sum_{\begin{subarray}{c} 1\le j\le n \\ j\not=i \end{subarray}} d_j^2.
\end{align}
Let us choose an integer $j_i\ge 0$ and write \eqref{N_dependence_on_d_i} in the form
\begin{align}
    N_{(d_1,\dotsc,d_n)} =d_i(-d_i +P_i-j_i) +d_i j_i -S_i +n-1.    
\end{align}
If now $d_i j_i -S_i +n-1\le 0$, we can conclude that
\begin{align}
    N_{(d_1,\dotsc,d_n)}=0, \ d_i\ge P_i-j_i.
    \label{N_equals_0}
\end{align}
The inequality $d_i j_i -S_i +n-1\le 0$ holds for $j_i=0$ (since $d_j\ge 1$ for all $1\le j\le n$ gives $S_i\ge n-1$), which implies that $N_{d_1,\dotsc,d_n}=0$ if at least one of the dimensions $d_i$ is at least as large as the product of all other dimensions $d_1,\dotsc,d_{i-1},d_{i+1},\dotsc,d_n$.
In some cases the stronger inequality \eqref{N_equals_0} with $j_i>0$ holds, but this requires additional constraints.
For example, for $n=4$ and $i=4$, the choice $j_i=1$ works for the cases $(d_1,d_2,d_3)\in\{(2,2,2),(2,2,3),(2,2,4),(2,3,3)\}$ investigated in Section \ref{section_examples_of_maximally_entangled_states} and the additional cases $(d_1,d_2,d_3)\in\{(2,3,4),(2,4,4)\}$, but it fails for all cases with $\min{(d_1,d_2,d_3)}\ge 3$. 

We also conclude from \eqref{N_conjecture} and \eqref{N_dependence_on_d_i} that for each $1\le i\le n$,
\begin{align}
    \max_{d_i}{N_{(d_1,\dotsc,d_n)}} &=
    \begin{cases}
        N_{(d_1,\dotsc,d_n)}\bigr\vert_{d_i=\frac{1}{2}P_i}, & P_i \textrm{\ is even}, \\
        N_{(d_1,\dotsc,d_n)}\bigr\vert_{d_i=\frac{1}{2}(P_i-1)} =N_{(d_1,\dotsc,d_n)}\bigr\vert_{d_i=\frac{1}{2}(P_i+1)}, & P_i \textrm{\ is odd}, \\
    \end{cases}
    \nn \\
        &=
    \begin{cases}
        \frac{1}{4}P_i^2 -S_i +n-1, & P_i \textrm{\ is even}, \\
        \frac{1}{4}(P_i^2-1) -S_i +n-1, & P_i \textrm{\ is odd}.
    \end{cases}
    \label{max_N}
\end{align}
All results in Section \ref{section_examples_of_maximally_entangled_states} agree with \eqref{N_conjecture_2_d2_d3}, \eqref{N_conjecture}, \eqref{recurrence_1}, \eqref{recurrence_k}, \eqref{N_dependence_on_d_i}, \eqref{N_equals_0} (for $j_i=1$) and \eqref{max_N}.

\section{Markov Chains}
\label{section_markov_chains}

It is convenient to represent the results of RWs in terms of the corresponding stochastic matrices.
For a set of classes $C_0,\dotsc,C_m$, the matrix element $P_{i,j}$ of the corresponding $(m+1)\times(m+1)$ stochastic matrix $P$ has the meaning of the transition probability from $C_j$ to $C_i$ in one time step.
The matrix $P$ satisfies $0\le P_{i,j}\le 1$, $0\le i,j\le m$ and the normalization condition for probabilities $\sum_{i=0}^m P_{i,j}=1$, $0\le j\le m$.
For each $k\in\N$, the matrix $P^k$ gives the transition probabilities after $k$ time steps, and it satisfies $0\le (P^k)_{i,j}\le 1$, $0\le i,j\le m$ and the normalization condition
\begin{align}
\sum_{i=0}^m (P^k)_{i,j}=1, \ k\in\N, \ 0\le j\le m.
\label{P_k_normalization}
\end{align}
The inductive proof of \eqref{P_k_normalization} uses $\sum_{i=0}^m P_{i,j}=1$ as the initial step and
\begin{align}
\sum_{i=0}^m (P^{k+1})_{i,j}=\sum_{i=0}^m \sum_{l=0}^m (P^k)_{i,l}P_{l,j}=\sum_{l=0}^m P_{l,j}=1
\end{align}
as the inductive step.

We assume that the Markov chain of transitions between the classes $C_0,\dotsc,C_m$ is an absorbing Markov chain with $C_m$ being the only absorbing class.
This means that the only transitions from $C_m$ are to $C_m$, and that each class among $C_0,\dotsc,C_{m-1}$ is a transient class, so that each $C_j$, $0\le j\le m-1$ can reach $C_m$ in a finite number of steps.
In terms of the stochastic matrix, this means that $P_{i,m}=\delta_{i,m}$ for each $0\le i\le m$, and that for each $0\le j\le m$ there exists the smallest $k_j\in\N_0$ such that $(P^{k_j})_{m,j}>0$.
The integer $k_j$ equals the smallest number of steps needed to go from the class $C_j$ to the class $C_m$.
(If such a $k_j$ does not exist, then $C_j$ is not a transient class, which is a contradiction.)
Clearly, $k_m=0$.

Let us consider the expected value of the number of steps $Q_j$ that it takes for RWs to transition from $C_j$ to $C_m$.
One expression for it is
\begin{align}
Q_j =\sum_{k=1}^\infty k (P^k - P^{k-1})_{m,j},
\label{Q_j_form_1}
\end{align}
where we set $P^0=I$ (the identity matrix) even if $P$ is singular.
Here $(P^k - P^{k-1})_{m,j}$ is the probability of getting from $C_j$ to $C_m$ in $k$ steps without getting there in $k-1$ steps, that is the transition takes exactly $k$ steps.
We can also compute $Q_j$ as the sum of probabilities for the RW to visit all states except for $C_m$ after all possible numbers of steps, 
\begin{align}
Q_j =\sum_{l=0}^{m-1}\sum_{k=0}^\infty (P^k)_{l,j}.
\label{Q_j_form_2}
\end{align}
The forms \eqref{Q_j_form_1} and \eqref{Q_j_form_2} are equivalent in view of
\begin{align}
\sum_{k=1}^\infty k (P^k - P^{k-1})_{m,j} &=\sum_{k=1}^\infty k\biggl( 1-\sum_{l=0}^{m-1}(P^k)_{l,j} -1+\sum_{l=0}^{m-1}(P^{k-1})_{l,j} \biggr) \nn \\
&=\sum_{l=0}^{m-1}\sum_{k=1}^\infty k(-P^k+P^{k-1})_{l,j} =\sum_{l=0}^{m-1}\sum_{k=0}^\infty (P^k)_{l,j},
\label{}
\end{align}
where we used \eqref{P_k_normalization} in the first equality.

Recognizing the special role played by the absorbing class $C_m$, we write $P$ as the block matrix
\begin{align}
P=
\begin{pmatrix}
\tilde{P} & 0 \\
X & 1
\end{pmatrix}
,
\end{align}
where $\tilde{P}$ is the $m\times m$ stochastic matrix for the classes $\{C_j\}_{j=0}^{m-1}$, $X$ is the $1\times m$ matrix corresponding to transitions probabilities from $\{C_j\}_{j=0}^{m-1}$ to $C_m$, and $0$ is the $m\times 1$ matrix of zeroes.
Using \eqref{Q_j_form_1} and the identity
\begin{align}
P^k=
\begin{pmatrix}
\tilde{P}^k & 0 \\
X \sum_{l=0}^{k-1}\tilde{P}^l & 1
\end{pmatrix}
,
\end{align}
we find
\begin{align}
Q_j &=\sum_{l=0}^{m-1}\sum_{k=0}^\infty (\tilde{P}^k)_{l,j} =\sum_{l=0}^{m-1}((I-\tilde{P})^{-1})_{l,j}, \ 0\le j\le m-1, \label{Q_j_tilde_P} \\
Q_m &=0.
\label{Q_m_tilde_P}
\end{align}
Note that the matrix $(I-P)^{-1}$ does not exists because $P$ has an eigenvalue $1$ as a result of $C_m$ being the absorbing class. 

From the data in Figure \eqref{figure_traffic_pattern_tripartite} we find
\begin{align}
P=
\begin{pmatrix}
0 & 0 & 0 & 0 & 0 & 0 & 0 \\
1 & 0.30 & 0 & 0 & 0 & 0 & 0 \\
0 & 0.17 & 0.24 & 0 & 0 & 0 & 0 \\
0 & 0.17 & 0 & 0.24 & 0 & 0 & 0 \\
0 & 0.17 & 0 & 0 & 0.24 & 0 & 0 \\
0 & 0 & 0.24 & 0.24 & 0.24 & 0.11 & 0 \\
0 & 0.20 & 0.52 & 0.52 & 0.52 & 0.89 & 1
\end{pmatrix}
, \  Q=
\begin{pmatrix}
3.63 \\
2.63 \\
1.68 \\
1.68 \\
1.68 \\
1.12 \\
0
\end{pmatrix}
, \ (I-\tilde{P})^{-1}=
\begin{pmatrix}
1 & 0 & 0 & 0 & 0 & 0 \\
1.43 & 1.43 & 0 & 0 & 0 & 0 \\
0.31 & 0.31 & 1.32 & 0 & 0 & 0 \\
0.31 & 0.31 & 0 & 1.32 & 0 & 0 \\
0.31 & 0.31 & 0 & 0 & 1.32 & 0 \\
0.26 & 0.26 & 0.36 & 0.36 & 0.36 & 1.13
\end{pmatrix}
.
\label{P_Q_2_2_2}
\end{align}
As in Figure \ref{figure_traffic_pattern_tripartite}, relations between elements of $P$, $Q$ and $(I-\tilde{P})^{-1}$ in \eqref{P_Q_2_2_2} hold exactly up to small variations due to rounding of numbers.
For each $0\le j\le m-1$, the number $Q_j$ coincides with the sum of the matrix elements of $(I-\tilde{P})^{-1}$ in the column $j$, in agreement with \eqref{Q_j_tilde_P}.
We check numerically that
\begin{align}
P^\infty=\lim_{k\to\infty} P^k =
\begin{pmatrix}
O & 0 \\
U & 1
\end{pmatrix}
,
\end{align}
where $O$ is the $m\times m $ matrix of zeroes and $U$ is the $1\times m$ matrix of ones, and $0$ is the $m\times 1$ matrix of zeroes.
The number $Q_0=3.63$ agrees with the corresponding value $3.6$ in Table \ref{table_MEC_properties_ddd}.

\begin{acknowledgments}
We thank Heinrich P\"as for a discussion about the early work of Zeh and others on decoherence.
The work of  TWK was supported by US DOE grant DE-SC0019235.
\end{acknowledgments}



\end{document}